\documentclass[pdflatex,sn-mathphys-num]{sn-jnl}

\usepackage{graphicx,epic,eepic,epsfig,amsmath,latexsym,amssymb,verbatim,color}
\usepackage{fullpage}

\usepackage{theorem}
\newtheorem{definition}{Definition}

\newtheorem{lemma}[definition]{Lemma}

\newtheorem{theorem}[definition]{Theorem}

\def\squareforqed{\hbox{\rlap{$\sqcap$}$\sqcup$}}
\def\qed{\ifmmode\squareforqed\else{\unskip\nobreak\hfil
\penalty50\hskip1em\null\nobreak\hfil\squareforqed
\parfillskip=0pt\finalhyphendemerits=0\endgraf}\fi}
\def\endenv{\ifmmode\;\else{\unskip\nobreak\hfil
\penalty50\hskip1em\null\nobreak\hfil\;
\parfillskip=0pt\finalhyphendemerits=0\endgraf}\fi}
\newenvironment{proof}{\noindent \textbf{{Proof~}}}{\qed}
\newenvironment{remark}{\noindent \textbf{{Remark~}}}

\mathchardef\ordinarycolon\mathcode`\:
\mathcode`\:=\string"8000
\def\vcentcolon{\mathrel{\mathop\ordinarycolon}}
\begingroup \catcode`\:=\active
  \lowercase{\endgroup
  \let :\vcentcolon
  }

\newcommand{\nc}{\newcommand}
\nc{\rnc}{\renewcommand}
\nc{\beq}{\begin{equation}}
\nc{\eeq}{{\end{equation}}}
\nc{\beqa}{\begin{eqnarray}}
\nc{\eeqa}{\end{eqnarray}}
\nc{\lbar}[1]{\overline{#1}}
\nc{\bra}[1]{\langle#1|}
\nc{\ket}[1]{|#1\rangle}
\nc{\ketbra}[2]{|#1\rangle\!\langle#2|}
\nc{\braket}[2]{\langle#1|#2\rangle}
\nc{\proj}[1]{| #1\rangle\!\langle #1 |}
\nc{\avg}[1]{\langle#1\rangle}
\nc{\supp}{\operatorname{supp}}
\nc{\rank}{\operatorname{rank}}
\nc{\smfrac}[2]{\mbox{$\frac{#1}{#2}$}}
\nc{\tr}{\operatorname{Tr}}
\nc{\ox}{\otimes}
\nc{\dg}{\dagger}
\nc{\dn}{\downarrow}
\nc{\cA}{{\cal A}}
\nc{\cB}{{\cal B}}
\nc{\cC}{{\cal C}}
\nc{\cD}{{\cal D}}
\nc{\cE}{{\cal E}}
\nc{\cF}{{\cal F}}
\nc{\cG}{{\cal G}}
\nc{\cH}{{\cal H}}
\nc{\cI}{{\cal I}}
\nc{\cJ}{{\cal J}}
\nc{\cK}{{\cal K}}
\nc{\cL}{{\cal L}}
\nc{\cM}{{\cal M}}
\nc{\cN}{{\cal N}}
\nc{\cO}{{\cal O}}
\nc{\cP}{{\cal P}}
\nc{\cQ}{{\cal Q}}
\nc{\cR}{{\cal R}}
\nc{\cS}{{\cal S}}
\nc{\cT}{{\cal T}}
\nc{\cU}{{\cal U}}
\nc{\cW}{{\cal W}}
\nc{\cX}{{\cal X}}
\nc{\cY}{{\cal Y}}
\nc{\cZ}{{\cal Z}}
\nc{\csupp}{{\operatorname{csupp}}}
\nc{\qsupp}{{\operatorname{qsupp}}}
\nc{\var}{{\operatorname{var}}}
\nc{\Var}{{\operatorname{Var}}}
\nc{\rar}{\rightarrow}
\nc{\lrar}{\longrightarrow}
\nc{\polylog}{{\operatorname{polylog}}}
\nc{\1}{{\openone}}
\nc{\wt}{{\operatorname{wt}}}
\nc{\av}[1]{{\left\langle {#1} \right\rangle}}

\newcommand\sA[0]{\mathrm{A}}
\newcommand\sB[0]{\mathrm{B}}
\newcommand\sT[0]{\mathrm{T}}

\nc{\RR}{{{\mathbb R}}}
\nc{\CC}{{{\mathbb C}}}
\nc{\FF}{{{\mathbb F}}}
\nc{\NN}{{{\mathbb N}}}
\nc{\ZZ}{{{\mathbb Z}}}
\nc{\PP}{{{\mathbb P}}}
\nc{\QQ}{{{\mathbb Q}}}
\nc{\UU}{{{\mathbb U}}}
\nc{\EE}{{{\mathbb E}}}
\nc{\id}{{\operatorname{id}}}

\nc{\CHSH}{{\operatorname{CHSH}}}

\nc{\be}{\begin{equation}}
\nc{\ee}{{\end{equation}}}
\nc{\bea}{\begin{eqnarray}}
\nc{\eea}{\end{eqnarray}}
\nc{\<}{\langle}
\rnc{\>}{\rangle}
\nc{\Hom}[2]{\mbox{Hom}(\CC^{#1},\CC^{#2})}
\nc{\rU}{\mbox{U}}

\nc{\ob}[1]{#1}

\nc{\tsep}{{\text{sep}}}
\nc{\SEP}{{\text{SEP}}}
\nc{\NS}{{\text{NS}}}
\nc{\LOCC}{{\text{LOCC}}}
\nc{\PPT}{{\text{PPT}}}
\nc{\EXT}{{\text{EXT}}}
\nc{\Sym}{{\operatorname{Sym}}}

\nc{\ERLO}{{E_{\text{r,LO}}}}
\nc{\ERLOCC}{{E_{\text{r,LOCC}}}}
\nc{\ERPPT}{{E_{\text{r,PPT}}}}
\nc{\ERLOCCinfty}{{E^{\infty}_{\text{r,LOCC}}}}

\newcommand{\Span}{{\operatorname{span}}}

\renewcommand{\theenumi}{(\roman{enumi})}

\renewcommand{\descriptionlabel}[1]{\hspace\labelsep \upshape\bfseries \textit{#1:}}

\begin{document}

\title{Geometry of entanglement and separability in Hilbert subspaces
of dimension up to three}

\author[1,2]{\fnm{Rotem} \sur{Liss}}

\author[1]{\fnm{Tal} \sur{Mor}}

\author*[3,4,5]{\fnm{Andreas} \sur{Winter}}\email{andreas.winter@uab.cat}

\affil[1]{\orgdiv{Computer Science Department}, \orgname{Technion},
\orgaddress{\city{Haifa}, \postcode{3200003}, \country{Israel}}}

\affil[2]{\orgdiv{ICFO---Institut de Ci\`encies Fot\`oniques},
\orgname{The Barcelona Institute of Science and Technology},
\orgaddress{\street{Av. Carl Friedrich Gauss 3}, \city{Castelldefels},
\postcode{08860}, \state{Barcelona}, \country{Spain}}}

\affil[3]{\orgdiv{ICREA {\&} F\'{\i}sica Te\`{o}rica: Informaci\'{o}
i Fen\`{o}mens Qu\`{a}ntics, Departament de F\'{\i}sica},
\orgname{Universitat Aut\`{o}noma de Barcelona},
\orgaddress{\city{Bellaterra}, \postcode{08193}, \state{Barcelona},
\country{Spain}}}

\affil[4]{\orgdiv{Institute for Advanced Study},
\orgname{Technische Universit\"at M\"unchen},
\orgaddress{\street{Lichtenbergstra{\ss}e 2a}, \city{Garching},
\postcode{85748}, \country{Germany}}}

\affil[5]{\orgname{QUIRCK---Quantum Information
Independent Research Centre Kessenich},
\orgaddress{\street{Gerhard-Samuel-Stra{\ss}e 14}, \city{Bonn},
\postcode{53129}, \country{Germany}}}

\abstract{We present a complete classification of the geometry of
the mutually complementary sets of
entangled and separable
states in three-dimensional Hilbert subspaces of bipartite and multipartite
quantum systems.
Our analysis begins by finding the geometric structure
of the pure product states in a given three-dimensional Hilbert subspace,
which determines all the possible separable and entangled mixed states
over the same subspace.
In bipartite systems, we characterise the $14$ possible
qualitatively different geometric shapes for the set of separable states
in any three-dimensional Hilbert subspace
($5$ classes which also appear in two-dimensional subspaces
and were found and analysed by Boyer, Liss and Mor 
[\emph{Phys. Rev. A} 95:032308, 2017 \cite{BLM}],
and $9$ novel classes which appear only in three-dimensional subspaces),
describe their geometries, and provide figures illustrating them.
We also generalise these results to characterise the sets of fully separable
states
(and hence the complementary sets of somewhat entangled states)
in three-dimensional subspaces of multipartite systems.
Our results show which geometrical forms quantum entanglement
can and cannot take in low-dimensional subspaces.}



\maketitle

\section{Introduction}
\label{sec:intro}
In the 89 years since the discovery of quantum entanglement~\cite{EPR}
and the realisation that it marks the main departure of quantum physics from any
form of classical explanation~\cite{Schroedinger1,Schroedinger2,Schroedinger3},
this quantum effect has been promoted
from a counter-intuitive foundational phenomenon~\cite{Bell,CHSH}
to the fuel of emerging
quantum technologies~\cite{BDSW,Ekert:QKD,JozsaLinden:ent,QTex}.
Indeed, quantum entanglement remains to this day the most studied
and most crucial feature separating non-classical information processing
from its classical counterpart~\cite{HHHH:ent}.

When classifying quantum entanglement, a major problem is deciding whether
a specific quantum state presents quantum entanglement (is ``entangled'') or
does not present quantum entanglement (is ``separable'' or ``non-entangled'').
It is well known that this general problem
is computationally hard (NP-hard)~\cite{Gurvits1,Gurvits2,Gharibian}:
its hardness results from the fact that quantum separability
is a bilinear condition which, if it were easy to solve, would allow the
encoding of quadratic constraints into otherwise convex (semidefinite) problems,
which is known to give rise to computationally hard problems~\cite{BeigiShor,HarrowMontanaro}.
However, the known hardness results only apply to the most general quantum states,
and in fact, it requires the consideration of highly mixed quantum 
states whose density matrices have maximal ranks.
For example, any pure state (i.e., a rank-one quantum state)
is known to be entangled if and only if its partial trace (reduced state)
is not pure, a condition easy to check. 
Formally, a pure state $\rho^{\sA\sB} = \proj{\psi}^{\sA\sB}$ 
is separable if and only if its partial trace 
$\rho^\sA \triangleq \tr_\sB \rho^{\sA\sB}$ is pure,
which happens if and only if $\tr (\rho^\sA)^2 = \tr \rho^\sA = 1$.

For quantum states whose rank is higher but still relatively small
(compared to the dimension of the full Hilbert space $\cH_\sA \ox \cH_\sB$),
analysing the reduced state is not enough, but separability
can still be efficiently detected in some cases.
First of all, if $\rank\rho^{\sA\sB}
< \max\left\{\rank\rho^\sA, \rank\rho^\sB\right\}$
then the state $\rho^{\sA\sB}$ is necessarily entangled
and in fact has distillable entanglement~\cite{HSTT99}.
If equality holds ($\rank\rho^{\sA\sB}
= \max\left\{\rank\rho^\sA, \rank\rho^\sB\right\}$),
or alternatively if $(\rank\rho^\sA)(\rank\rho^\sB) \leq 6$,
then the state $\rho^{\sA\sB}$ is separable if and only if
it has a ``positive partial transpose'' (PPT), namely, if applying
the partial transpose operator $(\cdot)^{\sT^\sB}$
to the density matrix $\rho^{\sA\sB}$
results in a positive semidefinite matrix~\cite{Peres,HHH:PPT,HLVC00}.
These results cover all bipartite density matrices of rank up to three.
If $\rank\rho^{\sA\sB} = 4$,
the result of Chen and Djokovi\'c~\cite{CD11} shows that
$\rho^{\sA\sB}$ is separable if and only if it is PPT
and its support includes a non-zero product state.
Furthermore, for all ranks bounded away from the maximum
(concretely, for all ranks $r \leq (\dim \cH_\sA - 1)(\dim \cH_\sB - 1)$),
the separable density matrices of rank $r$ form a set
of measure zero within the set of \emph{all} density matrices
of rank $r$~\cite{Lockhart}.

However, while deciding whether a \emph{specific} low-rank quantum state
is entangled or separable is an easy task, characterising the set of \emph{all}
separable states within a specific Hilbert subspace can be much harder.
(This is equivalent to analysing the set of separable states inside the
\emph{support} of a specific mixed state $\rho^{\sA\sB}$---namely,
inside the Hilbert subspace spanned by the eigenstates of $\rho^{\sA\sB}$
corresponding to non-zero eigenvalues.)
To achieve this goal, the first step requires classifying
the different Hilbert subspaces
(and the corresponding mixed states supported on them)
according to the geometric picture generated
by the set of all separable states inside each subspace.

The analysis for two-dimensional Hilbert subspaces
(Hilbert subspaces $S \subseteq \cH_\sA \ox \cH_\sB$ such that $\dim S = 2$),
and equivalently for rank-2 mixed states,
was completed by Boyer and two of the present authors~\cite{BLM}:
they found five (one plus one plus three) classes of two-dimensional subspaces
corresponding to five qualitatively different geometries
of their sets of separable states.
Generally speaking, $S$ may include \emph{no product states} at all;
may include exactly \emph{one} product state;
or may include \emph{two} linearly independent product states
in one of three qualitatively different ways
yielding three geometrically different sets of separable states,
detailed in Theorem~\ref{thm:BLM} of this paper.

In the present paper we push this line of investigation further by considering 
the case of three-dimensional Hilbert subspaces
(Hilbert subspaces $S \subseteq \cH_\sA \ox \cH_\sB$ such that $\dim S = 3$),
and equivalently of rank-3 mixed states.
We stress that in the present work,
similarly to the previous work~\cite{BLM}, only
the dimension of the subspace $S$ is fixed (to $3$),
while the local Hilbert spaces
$\cH_\sA$ and $\cH_\sB$ can have any dimension.
We prove that there are 14 geometrically distinct classes 
of such subspaces, which are the previous 5 classes of Theorem~\ref{thm:BLM}
(corresponding to~\cite{BLM})
in case $S$ includes at most two linearly independent product states, and 9
new classes (in case $S$ includes three linearly independent product states)
with different associated geometries of the sets of separable states. 
These classes are most naturally described using
the dimensions of the projections (partial traces) of the separable states onto
$\cH_\sA$ and $\cH_\sB$, whose supports we will denote by
$A_\tsep$ and $B_\tsep$, respectively.

The rest of this paper is organised as follows:
in Section~\ref{sec:setting} we define the mathematical setting.
In Section~\ref{sec:prev} we review the previous results of~\cite{BLM}
on two-dimensional subspaces,
and in Section~\ref{sec:main} we state and prove
our main result as Theorem~\ref{thm:main},
where the proof is divided according to the dimensions of the local projections.
Section~\ref{sec:illustrations} is dedicated to geometric
descriptions of the sets of separable states described in
Theorem~\ref{thm:main}, including visualisations of the possible classes.
In Section~\ref{sec:multipartite} we extend the results to multipartite systems.
Finally, in Section~\ref{sec:discussion} we conclude and discuss
the obtained results.

\section{Setting}
\label{sec:setting}
We begin the analysis given a rank-3 quantum mixed state $\rho^{\sA\sB}$ on
a bipartite system $\cH_\sA \ox \cH_\sB$. The \emph{support} of $\rho^{\sA\sB}$
is the three-dimensional Hilbert subspace $S = \supp \rho^{\sA\sB}$
which is spanned by the eigenstates of $\rho^{\sA\sB}$; that is,
if we write $\rho^{\sA\sB} = p_1 \proj{\psi_1}^{\sA\sB}
+ p_2 \proj{\psi_2}^{\sA\sB} + p_3 \proj{\psi_3}^{\sA\sB}$,
then $S = \supp \rho^{\sA\sB} = \Span\left\{\ket{\psi_1}^{\sA\sB},
\ket{\psi_2}^{\sA\sB}, \ket{\psi_3}^{\sA\sB}\right\}$.
We point out that, in general, $\supp \rho = (\ker \rho)^\perp$ and
$\rank \rho = \dim(\supp \rho)$ (which in our case equals $3$);
moreover, $S = \supp \rho$ includes all pure states
appearing in any of the possible decompositions
of the mixed state $\rho$ (see details and proof in Lemma~3 of~\cite{BLM}),
so the definition of $S$ is independent of the specific decomposition
chosen for $\rho$.

Formally, given a bipartite Hilbert space $\cH_\sA \ox \cH_\sB$
and given any Hilbert subspace $S \subseteq \cH_\sA \ox \cH_\sB$,
we denote by $S_\tsep$ the subspace spanned by all product states in $S$:
\begin{eqnarray}
S_\tsep &\triangleq& \Span \left\{\ket{\psi}^{\sA\sB} \in S ~ : ~
\ket{\psi}^{\sA\sB}\text{ is a product state}\right\} \nonumber \\
&=& \Span \left\{\ket{\psi}^{\sA\sB} \in S ~ : ~
\exists \ket{\phi_\sA}^\sA \in \cH_\sA ~ , ~ \ket{\phi_\sB}^\sB \in \cH_\sB ~ : ~
\ket{\psi}^{\sA\sB} = \ket{\phi_\sA}^\sA \ket{\phi_\sB}^\sB\right\}.
\end{eqnarray}
In our paper, we characterise the set of all separable states
(both pure and mixed) over $S$, a set we denote by $\cD^S_\tsep$.
This set includes exactly all mixtures (convex combinations)
of product states in $S$. Formally, we define:
\begin{equation}
\cD^S_\tsep \triangleq \{\rho \in \cD(S) ~ : ~ \rho\text{ is separable}\},
\end{equation}
where $\cD(S)$ is the set of
all positive semidefinite operators that have trace $1$
on the Hilbert space $S$.

Because $\dim S = 3$, the dimension of its subspace $S_\tsep$ can be
$0$, $1$, $2$, or $3$~\footnote{Note that if $S_\tsep$ is $0$-,
$1$-, or $2$-dimensional, it follows that $\rho^{\sA\sB}$
is necessarily entangled: if $\rho^{\sA\sB}$ is separable,
it must be a mixture of at least three linearly independent
product states from $S_\tsep$ (because $\rank \rho^{\sA\sB} = 3$),
so $\dim S_\tsep = 3$.}.
The analysis of the first three cases ($\dim S_\tsep \in \{0, 1, 2\}$)
was carried out in~\cite{BLM}~\footnote{All three cases
$\dim S_\tsep \in \{0, 1, 2\}$ exist,
because any given space $S_\tsep$ (spanned by product states)
can be extended by a suitable set of entangled states
that are linearly independent of $S_\tsep$ and of each other,
in order to reach the dimension $\dim S = 3$.}.
We therefore focus here on the last case $\dim S_\tsep = 3$, in which $S_\tsep = S$
is spanned by three linearly independent product states:
\begin{equation}\label{eq:S}
S = \Span\left\{\ket{\alpha_1}^\sA \ket{\beta_1}^\sB,
\ket{\alpha_2}^\sA \ket{\beta_2}^\sB,
\ket{\alpha_3}^\sA \ket{\beta_3}^\sB\right\}.
\end{equation}
Now, the classification will depend on the dimensionality
of the two \emph{local} subspaces
\begin{equation}\begin{split}\label{eq:Aprime-Bprime}
A_\tsep &\triangleq \Span\left\{\ket{\alpha_1}^\sA, \ket{\alpha_2}^\sA,
\ket{\alpha_3}^\sA\right\}, \\
B_\tsep &\triangleq \Span\left\{\ket{\beta_1}^\sB, \ket{\beta_2}^\sB,
\ket{\beta_3}^\sB\right\},
\end{split}\end{equation}
each of which can be $1$-, $2$-, or $3$-dimensional.
Writing all possible combinations of their dimensions
$(\dim A_\tsep,\dim B_\tsep)$, we obtain the following possible combinations:
(1,3) and (3,1), (3,3),
(2,3) and (3,2), and finally (2,2).
The resulting classes corresponding to each combination
are presented in Theorem~\ref{thm:main}.

\section{Previous work of~\cite{BLM}}
\label{sec:prev}
We begin by presenting the result of~\cite{BLM}
on a two-dimensional Hilbert subspace $S$,
which geometrically corresponds to a Bloch sphere.
This result is mainly based on distinguishing
the three simple cases (1,2), (2,1), and (2,2).

\begin{theorem}[Boyer/Liss/Mor~\cite{BLM}]
  \label{thm:BLM}
  Given a bipartite Hilbert space $\cH_\sA\ox \cH_\sB$
  and any two-dimensional subspace $S \subseteq \cH_\sA\ox \cH_\sB$,
  one of the 3 following cases holds:
  (note that each case corresponds to a possible dimension of the Hilbert
  subspace $S_\tsep$ spanned by all product states in $S$)
  \begin{enumerate}
    \item \label{thm:BLM0} $S$ includes no product states ($\dim S_\tsep = 0$),
    in which case all pure and mixed states over $S$ are entangled; or
    \item \label{thm:BLM1} $S$ includes exactly one (pure) product state
    ($\dim S_\tsep = 1$),
    in which case all the other pure and mixed states over $S$ are entangled; or
    \item \label{thm:BLM2} $S$ is spanned by two (pure) product states
    ($\dim S_\tsep = 2$), so:
    $S = \Span \left\{ \ket{\alpha_1}^\sA \ket{\beta_1}^\sB,
    \ket{\alpha_2}^\sA \ket{\beta_2}^\sB \right\}$.
  \end{enumerate}
  In the third case, let
  \begin{equation}\begin{split}
    A_\tsep &\triangleq \Span\left\{ \ket{\alpha_1}^\sA,
    \ket{\alpha_2}^\sA \right\}, \\
    B_\tsep &\triangleq \Span\left\{ \ket{\beta_1}^\sB,
    \ket{\beta_2}^\sB \right\},
  \end{split}\end{equation}
  whose dimensions can be
  $(\dim A_\tsep,\dim B_\tsep) \in \{(1,2),\,(2,1),\,(2,2)\}$.
  Then, the separability of the pure and mixed states over $S$ is as follows,
  depending on the dimensions $(\dim A_\tsep,\dim B_\tsep)$:
  \begin{description}
    \itemsep0em
    \item[(1,2)] $S = \left\{\ket{\alpha_1}^\sA\right\} \ox B_\tsep$ is simply
                 a local qubit in subsystem $\cH_\sB$,
                 and all pure and mixed states over $S$ are separable.
                 (In fact, all of them are product states
                 of a fixed $\ket{\alpha_1}^\sA$
                 and a qubit state over $B_\tsep$.)
    
    \item[(2,1)] $S = A_\tsep \ox \left\{\ket{\beta_1}^\sB\right\}$ is simply
                 a local qubit in subsystem $\cH_\sA$,
                 and all pure and mixed states over $S$ are separable.
                 (In fact, all of them are product states
                 of a qubit state over $A_\tsep$
                 and a fixed $\ket{\beta_1}^\sB$.)
    
    \item[(2,2)] The separable pure and mixed states over $S$
                 are exactly all mixtures (convex combinations)
                 of $\proj{\alpha_1}^\sA \ox \proj{\beta_1}^\sB$ and
                 $\proj{\alpha_2}^\sA \ox \proj{\beta_2}^\sB$;
                 all the other pure and mixed states over $S$ are entangled.
  \end{description}
\end{theorem}
\begin{proof}
Proved in~\cite{BLM} as its Theorem~9, which was phrased in different terms:
our case~\ref{thm:BLM0} corresponds to Class~5 of~\cite{BLM},
our case~\ref{thm:BLM1} corresponds to Class~4 of~\cite{BLM},
and in our case~\ref{thm:BLM2}, options~(1,2) and~(2,1)
correspond to Class~1 of~\cite{BLM},
and option~(2,2) corresponds to Class~2+3 of~\cite{BLM}.
\end{proof}

\medskip
\begin{remark}
We point out that~\cite{BLM} combined
both options $(\dim A_\tsep,\dim B_\tsep) \in \{(1,2),\,(2,1)\}$
into one class (Class~1), because they give the same final result
(all pure and mixed over $S$ being product states)
and are related to one another by the exchange symmetry between
$\cH_\sA$ and $\cH_\sB$. Here we distinguish between these two options,
so that classification becomes easier in the more involved case
of a three-dimensional Hilbert subspace (Theorem~\ref{thm:main}).

On the other hand, \cite{BLM} split option $(\dim A_\tsep,\dim B_\tsep) = (2,2)$
into two classes: Class~2 where the two states are orthogonal
($\ket{\alpha_1}^\sA \ket{\beta_1}^\sB \perp
\ket{\alpha_2}^\sA \ket{\beta_2}^\sB$)
and Class~3 where they are non-orthogonal.
Here we drop this distinction because it disappears
under the symmetry of local invertible operations.
As we explain below in Section~\ref{sec:main},
for invertible $X$ and $Y$ acting on $\cH_\sA$ and $\cH_\sB$, respectively,
$\rho^{\sA\sB}$ is separable
if and only if $(X \ox Y) \rho^{\sA\sB} (X \ox Y)^\dagger$ is separable.
We can thus use $X$ (or $Y$) to map any linearly independent set of states
in $\cH_\sA$ (or $\cH_\sB$, respectively)
to any other linearly independent set of states,
including in particular an orthonormal basis.
Classes~2 and~3 of~\cite{BLM} are thus the same under this symmetry,
so we cannot distinguish between them here.
\end{remark}

\section{Main result}
\label{sec:main}
We can now state our main result, modeled after Theorem~\ref{thm:BLM}.
The theorem uses the notations defined in Section~\ref{sec:setting}.

We note that $n$ vectors in a vector space of dimension $d$
are said to be ``in general position'' if any subset of size $d$ of them
is linearly independent.
In addition, $\cD(S)$ is the set of all density matrices
over the Hilbert space $S$,
and $\operatorname{conv}(A)$ is the set of all convex combinations (mixtures)
of the states in the set $A$.
Moreover, in the proof, quantum states are said to be ``equal''
even if they differ by a physically-irrelevant global phase $e^{i \varphi}$,
and they are always assumed to be normalised.

In the proof we also use the important observation that
separability of a state is invariant under local invertible operations.
In other words, for any two invertible matrices $X \in \text{GL}(\cH_\sA)$ and
$Y\in\text{GL}(\cH_\sB)$, the state $\rho^{\sA\sB}$ is separable if and only if
the state $(X \ox Y) \rho^{\sA\sB} (X \ox Y)^\dagger$ is separable.
Under this symmetry, the support is mapped accordingly:
if $S = \supp \rho^{\sA\sB}$,
then $\supp [(X \ox Y) \rho^{\sA\sB} (X \ox Y)^\dagger] = (X \ox Y)S$.
We will thus study the different shapes the set of separable states
$\cD^S_\tsep$ can take, where any two subspaces related
by a local invertible map $X \ox Y$ can be treated as equivalent.

\begin{theorem}
  \label{thm:main}
  Given a bipartite Hilbert space $\cH_\sA\ox \cH_\sB$
  and any three-dimensional subspace $S \subseteq \cH_\sA\ox \cH_\sB$,
  let $S_\tsep$ be the subspace spanned by all product states in $S$
  (that is, $S_\tsep \triangleq \Span \{\ket{\psi}^{\sA\sB} \in S ~ :
  \ket{\psi}^{\sA\sB}\text{ is a product state}\}$).
  We would like to characterise the set of all separable states
  (both pure and mixed) over $S$, a set we denote by $\cD^S_\tsep$
  and formally define as follows:
  \begin{equation}
  \cD^S_\tsep \triangleq \{\rho \in \cD(S) ~ : ~ \rho\text{ is separable}\}.
  \end{equation}
  
  If $\dim S_\tsep \leq 2$, then $\cD^S_\tsep$
  belongs to one of the 5 classes described in Theorem~\ref{thm:BLM}.
  
  If $\dim S_\tsep = 3$ (so $S_\tsep=S$),
  then using the notations of Eqs.~\eqref{eq:S}--\eqref{eq:Aprime-Bprime}:
  \begin{equation}\begin{split}
  S &= \Span\left\{\ket{\alpha_1}^\sA \ket{\beta_1}^\sB,
  \ket{\alpha_2}^\sA \ket{\beta_2}^\sB,
  \ket{\alpha_3}^\sA \ket{\beta_3}^\sB\right\}, \\
  A_\tsep &\triangleq \Span\left\{\ket{\alpha_1}^\sA, \ket{\alpha_2}^\sA,
  \ket{\alpha_3}^\sA\right\}, \\
  B_\tsep &\triangleq \Span\left\{\ket{\beta_1}^\sB, \ket{\beta_2}^\sB,
  \ket{\beta_3}^\sB\right\},
  \end{split}\end{equation}
  one of the 9 following cases occurs,
  depending on the dimensions $(\dim A_\tsep,\dim B_\tsep)$:

  \begin{description}
    \itemsep0em
    \item[(1,3)] All pure and mixed states over $S$ are separable, because
       $S = \left\{\ket{\alpha_1}^\sA\right\} \ox B_\tsep$.
       (In fact, all states over $S$ are product states
       of a fixed state $\ket{\alpha_1}^\sA$
       and a qutrit state over $B_\tsep$.)
       Formally, $\cD^S_\tsep = \cD(S)$.
       This case is illustrated in Fig.~\ref{fig:qutrit}.
  
    \item[(3,1)] This case is symmetric to (1,3):
       all pure and mixed states over $S$ are separable, because
       $S = A_\tsep \ox \left\{\ket{\beta_1}^\sB\right\}$.
       (In fact, all states over $S$ are product states
       of a qutrit state over $A_\tsep$
       and a fixed state $\ket{\beta_1}^\sB$.)
       Formally, $\cD^S_\tsep = \cD(S)$.
       This case is illustrated in Fig.~\ref{fig:qutrit}.
  
    \item[(3,3)] The separable pure and mixed states over $S$
       are exactly all mixtures (convex combinations)
       of $\proj{\alpha_1}^\sA \ox \proj{\beta_1}^\sB$,
       $\proj{\alpha_2}^\sA \ox \proj{\beta_2}^\sB$,
       and $\proj{\alpha_3}^\sA \ox \proj{\beta_3}^\sB$;
       all the other pure and mixed states over $S$ are entangled.
       Formally:
       \begin{equation}
       \cD^S_\tsep = \operatorname{conv}
       \left\{\proj{\alpha_1}^\sA \ox \proj{\beta_1}^\sB ~ , ~
       \proj{\alpha_2}^\sA \ox \proj{\beta_2}^\sB ~ , ~
       \proj{\alpha_3}^\sA \ox \proj{\beta_3}^\sB\right\}.
       \end{equation}
       This case is illustrated in Fig.~\ref{fig:ctrit}.

    \item[(2,3)] We distinguish two subcases:
       \begin{enumerate}
        \item If the three states
        $\ket{\alpha_1}^\sA, \ket{\alpha_2}^\sA, \ket{\alpha_3}^\sA$
        are linearly dependent but in general position
        (i.e., any two of them are linearly independent),
        then the separable pure and mixed states over $S$
        are exactly all mixtures (convex combinations) of
        $\proj{\alpha_1}^\sA \ox \proj{\beta_1}^\sB$, 
        $\proj{\alpha_2}^\sA \ox \proj{\beta_2}^\sB$, and
        $\proj{\alpha_3}^\sA \ox \proj{\beta_3}^\sB$,
        identically to case (3,3). Formally:
        \begin{equation}
        \cD^S_\tsep = \operatorname{conv}
        \left\{\proj{\alpha_1}^\sA \ox \proj{\beta_1}^\sB ~ , ~
        \proj{\alpha_2}^\sA \ox \proj{\beta_2}^\sB ~ , ~
        \proj{\alpha_3}^\sA \ox \proj{\beta_3}^\sB\right\}.
        \end{equation}
        This case is illustrated in Fig.~\ref{fig:ctrit}.
        \item Otherwise, without loss of generality
        we can assume $\ket{\alpha_2}^\sA = \ket{\alpha_3}^\sA$,
        and the separable pure and mixed states over $S$
        are exactly all mixtures (convex combinations) of
        $\proj{\alpha_1}^\sA \ox \proj{\beta_1}^\sB$
        with any pure or mixed state over the space
        $\left\{\ket{\alpha_2}^\sA\right\}
        \ox \Span\left\{\ket{\beta_2}^\sB, \ket{\beta_3}^\sB\right\}$.
        Formally, in case $\ket{\alpha_2}^\sA = \ket{\alpha_3}^\sA$:
        \begin{equation}
        \cD^S_\tsep = \operatorname{conv}
        \left[\left\{\proj{\alpha_1}^\sA \ox \proj{\beta_1}^\sB\right\} \cup
        \cD\left(\left\{\ket{\alpha_2}^\sA\right\} \ox
        \Span\left\{\ket{\beta_2}^\sB, \ket{\beta_3}^\sB\right\}\right)\right],
        \end{equation}
        and symmetric results are obtained
        in case $\ket{\alpha_1}^\sA = \ket{\alpha_2}^\sA$
        or $\ket{\alpha_1}^\sA = \ket{\alpha_3}^\sA$.
        This case is illustrated in Fig.~\ref{fig:qubit-cone}.
       \end{enumerate}

    \item[(3,2)] This case is symmetric to (2,3):
       we distinguish two subcases:
       \begin{enumerate}
        \item If the three states
        $\ket{\beta_1}^\sB, \ket{\beta_2}^\sB, \ket{\beta_3}^\sB$
        are linearly dependent but in general position
        (i.e., any two of them are linearly independent),
        then the separable pure and mixed states over $S$
        are exactly all mixtures (convex combinations) of
        $\proj{\alpha_1}^\sA \ox \proj{\beta_1}^\sB$, 
        $\proj{\alpha_2}^\sA \ox \proj{\beta_2}^\sB$, and
        $\proj{\alpha_3}^\sA \ox \proj{\beta_3}^\sB$,
        identically to case (3,3). Formally:
        \begin{equation}
        \cD^S_\tsep = \operatorname{conv}
        \left\{\proj{\alpha_1}^\sA \ox \proj{\beta_1}^\sB ~ , ~
        \proj{\alpha_2}^\sA \ox \proj{\beta_2}^\sB ~ , ~
        \proj{\alpha_3}^\sA \ox \proj{\beta_3}^\sB\right\}.
        \end{equation}
        This case is illustrated in Fig.~\ref{fig:ctrit}.
        \item Otherwise, without loss of generality
        we can assume $\ket{\beta_2}^\sB = \ket{\beta_3}^\sB$,
        and the separable pure and mixed states over $S$
        are exactly all mixtures (convex combinations) of
        $\proj{\alpha_1}^\sA \ox \proj{\beta_1}^\sB$
        with any pure or mixed state over the space
        $\Span\left\{\ket{\alpha_2}^\sA, \ket{\alpha_3}^\sA\right\}
        \ox \left\{\ket{\beta_2}^\sB\right\}$.
        Formally, in case $\ket{\beta_2}^\sB = \ket{\beta_3}^\sB$:
        \begin{equation}
        \cD^S_\tsep = \operatorname{conv}
        \left[\left\{\proj{\alpha_1}^\sA \ox \proj{\beta_1}^\sB\right\} \cup
        \cD\left(\Span\left\{\ket{\alpha_2}^\sA, \ket{\alpha_3}^\sA\right\} \ox
        \left\{\ket{\beta_2}^\sB\right\}\right)\right],
        \end{equation}
        and symmetric results are obtained
        in case $\ket{\beta_1}^\sB = \ket{\beta_2}^\sB$
        or $\ket{\beta_1}^\sB = \ket{\beta_3}^\sB$.
        This case is illustrated in Fig.~\ref{fig:qubit-cone}.
       \end{enumerate}

    \item[(2,2)] We distinguish two subcases:
       \begin{enumerate}
        \item If the three states 
        $\ket{\alpha_1}^\sA, \ket{\alpha_2}^\sA, \ket{\alpha_3}^\sA$
        are not in general position,
        or symmetrically if the three states
        $\ket{\beta_1}^\sB, \ket{\beta_2}^\sB, \ket{\beta_3}^\sB$
        are not in general position,
        then without loss of generality we can assume
        $\ket{\alpha_2}^\sA = \ket{\alpha_3}^\sA$
        and $\ket{\beta_1}^\sB \ne \ket{\beta_2}^\sB$,
        and then the separable pure and mixed states over $S$
        are exactly all mixtures (convex combinations)
        of any pure or mixed state over the subspace $S_1 \triangleq
        \left\{\ket{\alpha_2}^\sA\right\} \ox
        \Span\left\{\ket{\beta_1}^\sB, \ket{\beta_2}^\sB\right\}$
        with any pure or mixed state over the subspace $S_2 \triangleq
        \Span\left\{\ket{\alpha_1}^\sA, \ket{\alpha_2}^\sA\right\} \ox
        \left\{\ket{\beta_1}^\sB\right\}$.
        Formally, in case $\ket{\alpha_2}^\sA = \ket{\alpha_3}^\sA$
        (and $\ket{\beta_1}^\sB \ne \ket{\beta_2}^\sB$):
        \begin{eqnarray}
        \cD^S_\tsep &=& \operatorname{conv}\left[\cD(S_1) \cup \cD(S_2)\right]
        \nonumber \\
        &=& \operatorname{conv}\left[\cD\left(\left\{\ket{\alpha_2}^\sA\right\}
        \ox \Span\left\{\ket{\beta_1}^\sB, \ket{\beta_2}^\sB\right\}\right)
        \cup \cD\left(\Span\left\{\ket{\alpha_1}^\sA, \ket{\alpha_2}^\sA\right\}
        \ox \left\{\ket{\beta_1}^\sB\right\}\right)\right] \nonumber \\
        &=& \operatorname{conv}\left[\cD\left(\left\{\ket{\alpha_2}^\sA\right\}
        \ox B_\tsep\right)
        \cup \cD\left(A_\tsep \ox \left\{\ket{\beta_1}^\sB\right\}\right)\right],
        \end{eqnarray}
        and symmetric results are obtained
        in case $\ket{\alpha_1}^\sA = \ket{\alpha_2}^\sA$,
        $\ket{\alpha_1}^\sA = \ket{\alpha_3}^\sA$,
        $\ket{\beta_2}^\sB = \ket{\beta_3}^\sB$,
        $\ket{\beta_1}^\sB = \ket{\beta_2}^\sB$,
        or $\ket{\beta_1}^\sB = \ket{\beta_3}^\sB$.
        This case is illustrated in Fig.~\ref{fig:qubit+qubit}.
        \item If the three states
        $\ket{\alpha_1}^\sA, \ket{\alpha_2}^\sA, \ket{\alpha_3}^\sA$
        and the three states
        $\ket{\beta_1}^\sB, \ket{\beta_2}^\sB, \ket{\beta_3}^\sB$
        are both linearly dependent but in general position
        (i.e., any two states of $\ket{\alpha_1}^\sA, \ket{\alpha_2}^\sA,
        \ket{\alpha_3}^\sA$ are linearly independent, and similarly,
        any two states of $\ket{\beta_1}^\sB, \ket{\beta_2}^\sB,
        \ket{\beta_3}^\sB$ are linearly independent),
        then there exists an invertible linear map $L:A_\tsep \mapsto B_\tsep$
        satisfying $\ket{\beta_i}^\sB \propto L \ket{\alpha_i}^\sA$
        for all values $i=1,2,3$, such that the product states
        in $S$ are exactly $\left\{ \ket{\psi}^\sA
        \left(L\ket{\psi}\right)^\sB ~ : ~ \ket{\psi}^\sA \in A_\tsep \right\}$
        (up to a normalisation factor),
        and the separable pure and mixed states over $S$ are all mixtures
        of these product states in $S$. Formally:
        \begin{equation}
        \cD^S_\tsep = \operatorname{conv}
        \left\{\proj{\Psi}^{\sA\sB}~:~\ket{\Psi}^{\sA\sB}
        = \gamma \ket{\psi}^\sA \left(L \ket{\psi}\right)^\sB~,~
        \ket{\psi}^\sA \in A_\tsep~,~\gamma \in \mathbb{C}~,~
        \left\|\ket{\Psi}^{\sA\sB}\right\| = 1 \right\}.
        \end{equation}
        This case is illustrated in Fig.~\ref{fig:phi-x-phi}.
       \end{enumerate}
  \end{description}
\end{theorem}
\begin{proof}
If $\dim S_\tsep \leq 2$, then applying Theorem~\ref{thm:BLM}
to a two-dimensional subspace of $S$ which includes $S_\tsep$ proves that
the set $\cD^{S_\tsep}_\tsep$ of pure and mixed separable states over $S_\tsep$
belongs to one of the five classes described in Theorem~\ref{thm:BLM}.
All the other pure and mixed states over $S$ (that are not states over $S_\tsep$)
must be entangled. Therefore, $\cD^S_\tsep = \cD^{S_\tsep}_\tsep$ indeed
belongs to one of the five classes described in Theorem~\ref{thm:BLM}.

If $\dim S_\tsep = 3$, then $S_\tsep = S$,
in which case we divide our proof into cases
according to the dimensions $(\dim A_\tsep, \dim B_\tsep)$:

\subsection{\label{subsec:proof_13_31}(1,3) and (3,1)}
The case (1,3) means that $\dim A_\tsep = 1$, so $\ket{\alpha_1}^\sA
= \ket{\alpha_2}^\sA = \ket{\alpha_3}^\sA$. Therefore,
\begin{equation}\begin{split}
S &= \Span\left\{\ket{\alpha_1}^\sA \ket{\beta_1}^\sB,
\ket{\alpha_2}^\sA \ket{\beta_2}^\sB,
\ket{\alpha_3}^\sA \ket{\beta_3}^\sB\right\} \\
&= \left\{\ket{\alpha_1}^\sA\right\}
\ox \Span\left\{\ket{\beta_1}^\sB, \ket{\beta_2}^\sB, \ket{\beta_3}^\sB\right\}
= \left\{\ket{\alpha_1}^\sA\right\} \ox B_\tsep,
\end{split}\end{equation}
which means that all pure and mixed states over $S$ are product states:
$\cD^S_\tsep = \cD(S)$, as needed.

The proof for the case (3,1) is symmetric,
so $S = A_\tsep \ox \left\{\ket{\beta_1}^\sB\right\}$ and $\cD^S_\tsep = \cD(S)$,
as needed.

\subsection{\label{subsec:proof_33}(3,3)}
In this case, all three states $\ket{\alpha_1}^\sA, \ket{\alpha_2}^\sA,
\ket{\alpha_3}^\sA$ are linearly independent, and so are $\ket{\beta_1}^\sB,
\ket{\beta_2}^\sB, \ket{\beta_3}^\sB$.
Let $\left\{\ket{1}^\sA, \ket{2}^\sA, \ket{3}^\sA\right\}$
and $\left\{\ket{1}^\sB, \ket{2}^\sB, \ket{3}^\sB\right\}$
be orthonormal bases of $A_\tsep$ and $B_\tsep$, respectively.
We can thus apply a local invertible map
$X \ox Y$, where $X$ is an invertible matrix over $A_\tsep$ mapping
$\ket{\alpha_1}^\sA \mapsto \ket{1}^\sA$,
$\ket{\alpha_2}^\sA \mapsto \ket{2}^\sA$, and
$\ket{\alpha_3}^\sA \mapsto \ket{3}^\sA$,
while $Y$ is an invertible matrix over $B_\tsep$ mapping
$\ket{\beta_1}^\sB \mapsto \ket{1}^\sB$,
$\ket{\beta_2}^\sB \mapsto \ket{2}^\sB$, and
$\ket{\beta_3}^\sB \mapsto \ket{3}^\sB$
(both maps exist, because there always exists an invertible matrix
mapping a given set of linearly independent states
into another given set of linearly independent states).

Since $S = \Span \left\{\ket{\alpha_1}^\sA \ket{\beta_1}^\sB,
\ket{\alpha_2}^\sA \ket{\beta_2}^\sB, \ket{\alpha_3}^\sA
\ket{\beta_3}^\sB\right\}$,
all elements of $S$ are superpositions
$u\ket{\alpha_1}^\sA \ket{\beta_1}^\sB + v\ket{\alpha_2}^\sA \ket{\beta_2}^\sB
+ w\ket{\alpha_3}^\sA \ket{\beta_3}^\sB$.
Under the above map, they are mapped to the respective superpositions
$u\ket{1}^\sA \ket{1}^\sB + v\ket{2}^\sA \ket{2}^\sB
+ w\ket{3}^\sA \ket{3}^\sB$.

Since any local invertible map preserves separability,
we notice that each state in $S$ is a product state if and only if
the state to which it is mapped is a product state.
However, since $u\ket{1}^\sA \ket{1}^\sB + v\ket{2}^\sA \ket{2}^\sB
+ w\ket{3}^\sA \ket{3}^\sB$ is in its Schmidt form
and has a Schmidt number of $3$ or less,
it is a product state if and only if two of its coefficients are zero
(otherwise, it is entangled).
Therefore, the only product states in $S$
are the spanning states $\ket{\alpha_1}^\sA \ket{\beta_1}^\sB,
\ket{\alpha_2}^\sA \ket{\beta_2}^\sB, \ket{\alpha_3}^\sA \ket{\beta_3}^\sB$.

This finding implies that the set of separable pure and mixed states over $S$
includes exactly all mixtures of these three product states, so:
\begin{equation}
\cD^S_\tsep = \operatorname{conv}
\left\{\proj{\alpha_1}^\sA \ox \proj{\beta_1}^\sB ~ , ~
\proj{\alpha_2}^\sA \ox \proj{\beta_2}^\sB ~ , ~
\proj{\alpha_3}^\sA \ox \proj{\beta_3}^\sB\right\},
\end{equation}
as needed.

\subsection{\label{subsec:proof_23_32}(2,3) and (3,2)}
In the case (2,3) we can assume,
without loss of generality, that $\ket{\alpha_1}^\sA \ne \ket{\alpha_2}^\sA$,
so $\ket{\alpha_3}^\sA = a\ket{\alpha_1}^\sA + b\ket{\alpha_2}^\sA$.
On the other hand, $\ket{\beta_1}^\sB, \ket{\beta_2}^\sB, \ket{\beta_3}^\sB$
are linearly independent.
Using again local invertible linear maps, we can map
$\ket{\alpha_1}^\sA \mapsto \ket{1}^\sA$ ~ , ~
$\ket{\alpha_2}^\sA \mapsto \ket{2}^\sA$ ~ , ~
$\ket{\beta_1}^\sB \mapsto \ket{1}^\sB$ ~ , ~
$\ket{\beta_2}^\sB \mapsto \ket{2}^\sB$, and
$\ket{\beta_3}^\sB \mapsto \ket{3}^\sB$,
which also maps $\ket{\alpha_3}^\sA \mapsto a\ket{1}^\sA + b\ket{2}^\sA$.
A general state in $S$ is then a superposition of the following form:
\begin{equation}\begin{split}
  \label{eq:2-3}
  u\ket{\alpha_1}^\sA \ket{\beta_1}^\sB + v\ket{\alpha_2}^\sA \ket{\beta_2}^\sB
            + w\ket{\alpha_3}^\sA\ket{\beta_3}^\sB
            &\mapsto u\ket{1}^\sA \ket{1}^\sB + v\ket{2}^\sA \ket{2}^\sB
            + wa\ket{1}^\sA \ket{3}^\sB + wb\ket{2}^\sA \ket{3}^\sB \\
            &=  \ket{1}^\sA \left(u\ket{1}^\sB + wa\ket{3}^\sB\right)
            + \ket{2}^\sA \left(v\ket{2}^\sB + wb\ket{3}^\sB\right).
\end{split}\end{equation}
This is a product state if and only if
$u\ket{1}^\sB + wa\ket{3}^\sB$ and $v\ket{2}^\sB + wb\ket{3}^\sB$ are
linearly dependent (including the possibility that one of them is $0$). 
There are thus two subcases:
\begin{enumerate}
\item If the three states $\ket{\alpha_1}^\sA, \ket{\alpha_2}^\sA,
     \ket{\alpha_3}^\sA$ are in general position (that is, any two of them
     are linearly independent), then in particular $a,b \neq 0$.
     This means that for $u\ket{1}^\sB + wa\ket{3}^\sB$ and
     $v\ket{2}^\sB + wb\ket{3}^\sB$ to be linearly dependent, there are only
     three possibilities: $u=v=0$, $u=w=0$, and $v=w=0$. Therefore,
     the only product states in $S$ are the spanning states
     $\ket{\alpha_1}^\sA \ket{\beta_1}^\sB,
     \ket{\alpha_2}^\sA \ket{\beta_2}^\sB,
     \ket{\alpha_3}^\sA \ket{\beta_3}^\sB$.
     This implies, identically to case (3,3), that the only separable states
     over $S$ are the mixtures of these three product states, so:
     \begin{equation}
     \cD^S_\tsep = \operatorname{conv}
     \left\{\proj{\alpha_1}^\sA \ox \proj{\beta_1}^\sB ~ , ~
     \proj{\alpha_2}^\sA \ox \proj{\beta_2}^\sB ~ , ~
     \proj{\alpha_3}^\sA \ox \proj{\beta_3}^\sB\right\},
     \end{equation}
     as needed.
\item Otherwise, without loss of generality we can assume
     $\ket{\alpha_2}^\sA=\ket{\alpha_3}^\sA$, so $a=0$ ~ , ~ $b=1$.
     In this case, the resulting states $u\ket{1}^\sB$ and
     $v\ket{2}^\sB + w\ket{3}^\sB$ are linearly dependent
     if and only if $u=0$ or $v=w=0$.
     Therefore, the product states in $S$ are either of the form
     $\ket{\alpha_2}^\sA \left(v\ket{\beta_2}^\sB + w\ket{\beta_3}^\sB\right)$
     (corresponding to case $u=0$)
     which are exactly all states in the Hilbert subspace
     $\left\{\ket{\alpha_2}^\sA\right\} \ox
     \Span\left\{ \ket{\beta_2}^\sB, \ket{\beta_3}^\sB \right\}$,
     or the single state $\ket{\alpha_1}^\sA \ket{\beta_1}^\sB$
     (corresponding to case $v=w=0$).
     This implies that the separable pure and mixed states over $S$
     are all mixtures of $\ket{\alpha_1}^\sA \ket{\beta_1}^\sB$
     with any pure or mixed state from $\left\{\ket{\alpha_2}^\sA\right\} \ox
     \Span\left\{ \ket{\beta_2}^\sB, \ket{\beta_3}^\sB \right\}$, so:
     \begin{equation}
     \cD^S_\tsep = \operatorname{conv}
     \left[\left\{\proj{\alpha_1}^\sA \ox \proj{\beta_1}^\sB\right\} \cup
     \cD\left(\left\{\ket{\alpha_2}^\sA\right\} \ox
     \Span\left\{\ket{\beta_2}^\sB, \ket{\beta_3}^\sB\right\}\right)\right],
     \end{equation}
     as needed.
\end{enumerate}
The proof for the case (3,2) is symmetric.

\subsection{\label{subsec:proof_22}(2,2)}
We can assume, without loss of generality,
that $\ket{\alpha_1}^\sA \ne \ket{\alpha_2}^\sA$
and $\ket{\beta_1}^\sB \ne \ket{\beta_2}^\sB$,
so $\ket{\alpha_3}^\sA = a\ket{\alpha_1}^\sA + b\ket{\alpha_2}^\sA$
and $\ket{\beta_3}^\sB = c\ket{\beta_1}^\sB + d\ket{\beta_2}^\sB$.
As before, using local invertible linear maps, we can
map $\ket{\alpha_1}^\sA \mapsto \ket{1}^\sA$ ~ , ~
$\ket{\alpha_2}^\sA \mapsto \ket{2}^\sA$ ~ , ~
$\ket{\beta_1}^\sB \mapsto \ket{1}^\sB$, and
$\ket{\beta_2}^\sB \mapsto \ket{2}^\sB$,
which also maps $\ket{\alpha_3}^\sA \mapsto a\ket{1}^\sA+b\ket{2}^\sA$
and $\ket{\beta_3}^\sB \mapsto c\ket{1}^\sB+d\ket{2}^\sB$.
We can thus write a general state in $S$ as a superposition:
\begin{equation}\begin{split}
  \label{eq:2-2}
  u\ket{\alpha_1}^\sA \ket{\beta_1}^\sB + v\ket{\alpha_2}^\sA \ket{\beta_2}^\sB
       + w\ket{\alpha_3}^\sA \ket{\beta_3}^\sB
       &\mapsto u\ket{1}^\sA \ket{1}^\sB + v\ket{2}^\sA \ket{2}^\sB \\
          &+ wac\ket{1}^\sA \ket{1}^\sB
          + wbd\ket{2}^\sA \ket{2}^\sB
          + wad\ket{1}^\sA \ket{2}^\sB + wbc\ket{2}^\sA \ket{1}^\sB \\
       &= (u+wac)\ket{1}^\sA \ket{1}^\sB + (v+wbd)\ket{2}^\sA \ket{2}^\sB \\
       &+ wad\ket{1}^\sA \ket{2}^\sB + wbc\ket{2}^\sA \ket{1}^\sB.
\end{split}\end{equation}
This is a product state if and only if the determinant
of the coefficients is $0$, i.e.
\begin{equation}\label{product_det}
  (u+wac)(v+wbd) = w^2 abcd.
\end{equation}
There are thus two subcases:
\begin{enumerate}
\item If either the three states $\ket{\alpha_1}^\sA, \ket{\alpha_2}^\sA,
     \ket{\alpha_3}^\sA$ or the three states $\ket{\beta_1}^\sB,
     \ket{\beta_2}^\sB, \ket{\beta_3}^\sB$ are not in general position
     (that is, if not every two of these states are linearly independent),
     then one of $a$, $b$, $c$, or $d$ is $0$. Without loss of generality,
     assume $a = 0$, so $\ket{\alpha_2}^\sA=\ket{\alpha_3}^\sA$.
     This means in particular that the following three states are in $S$:
     $\ket{\alpha_1}^\sA \ket{\beta_1}^\sB$,
     $\ket{\alpha_2}^\sA \ket{\beta_2}^\sB$, and
     $\ket{\alpha_3}^\sA \ket{\beta_3}^\sB =
     \ket{\alpha_2}^\sA \left(c\ket{\beta_1}^\sB+d\ket{\beta_2}^\sB\right)$,
     so by linearity we deduce that
     $\ket{\alpha_2}^\sA \ket{\beta_1}^\sB$ is also in $S$
     (note that $c \neq 0$, because otherwise the equality
     $\ket{\alpha_2}^\sA \ket{\beta_2}^\sB =
     \ket{\alpha_3}^\sA \ket{\beta_3}^\sB$ would hold,
     which would contradict the fact $\dim S = 3$). Thus:
     \begin{equation}
     S = \Span\left\{ \ket{\alpha_1}^\sA \ket{\beta_1}^\sB,
     \ket{\alpha_2}^\sA \ket{\beta_1}^\sB,
     \ket{\alpha_2}^\sA \ket{\beta_2}^\sB \right\}.
     \end{equation}
     $S$ is thus the span of a union of the two local qubit spaces
     $S_1 \triangleq \left\{\ket{\alpha_2}^\sA\right\}
     \ox \Span\left\{\ket{\beta_1}^\sB, \ket{\beta_2}^\sB\right\}$
     and $S_2 \triangleq \Span\left\{\ket{\alpha_1}^\sA,
     \ket{\alpha_2}^\sA\right\}
     \ox \left\{\ket{\beta_1}^\sB\right\}$,
     which intersect in $\ket{\alpha_2}^\sA \ket{\beta_1}^\sB$.
     Thus, the separable pure and mixed states over $S$
     are exactly all mixtures (convex combinations)
     of any pure or mixed state over the space $S_1$
     with any pure or mixed state over the space $S_2$. Formally:
     \begin{eqnarray}
     \cD^S_\tsep &=& \operatorname{conv}\left[\cD(S_1) \cup \cD(S_2)\right]
     \nonumber \\
     &=& \operatorname{conv}\left[\cD\left(\left\{\ket{\alpha_2}^\sA\right\} \ox
     \Span\left\{\ket{\beta_1}^\sB, \ket{\beta_2}^\sB\right\}\right) \cup
     \cD\left(\Span\left\{\ket{\alpha_1}^\sA, \ket{\alpha_2}^\sA\right\} \ox
     \left\{\ket{\beta_1}^\sB\right\}\right)\right] \nonumber \\
     &=& \operatorname{conv}\left[\cD\left(\left\{\ket{\alpha_2}^\sA\right\} \ox
     B_\tsep\right) \cup
     \cD\left(A_\tsep \ox \left\{\ket{\beta_1}^\sB\right\}\right)\right],
     \end{eqnarray}
     as needed.
 
\item Otherwise (if both $\ket{\alpha_1}^\sA, \ket{\alpha_2}^\sA,
     \ket{\alpha_3}^\sA$ and $\ket{\beta_1}^\sB, \ket{\beta_2}^\sB,
     \ket{\beta_3}^\sB$ are in general position),
     this means that $a,b,c,d \neq 0$.
     We can thus choose new variables $x \triangleq u+wac$ ~ , ~
     $y \triangleq v+wbd$, and $z \triangleq wbc$,
     so that Eq.~\eqref{product_det} becomes
     \begin{equation}\label{product_det_simp}
       xy = w^2 abcd = z^2\frac{ad}{bc},
     \end{equation}
     which means that all product states in $S$ are of the form:
     \begin{equation}\begin{split}
       \label{eq:prod-alt}
       u\ket{\alpha_1}^\sA \ket{\beta_1}^\sB
         + v\ket{\alpha_2}^\sA \ket{\beta_2}^\sB
         + w\ket{\alpha_3}^\sA \ket{\beta_3}^\sB
         &\mapsto x\ket{1}^\sA \ket{1}^\sB
         + y\ket{2}^\sA \ket{2}^\sB \\
         &+ z{\textstyle\frac{ad}{bc}}\ket{1}^\sA \ket{2}^\sB
         + z\ket{2}^\sA \ket{1}^\sB \\
         &\propto x^2\ket{1}^\sA \ket{1}^\sB
         + xy\ket{2}^\sA \ket{2}^\sB \\
         &+ xz{\textstyle\frac{ad}{bc}}\ket{1}^\sA \ket{2}^\sB
         + xz\ket{2}^\sA \ket{1}^\sB \\
         &= x^2\ket{1}^\sA \ket{1}^\sB
         + z^2{\textstyle\frac{ad}{bc}}\ket{2}^\sA \ket{2}^\sB \\
         &+ xz{\textstyle\frac{ad}{bc}}\ket{1}^\sA \ket{2}^\sB
         + xz\ket{2}^\sA \ket{1}^\sB \\
         &= \left(x\ket{1}^\sA + z\ket{2}^\sA\right)
         \left(x\ket{1}^\sB + z{\textstyle\frac{ad}{bc}}\ket{2}^\sB\right).
     \end{split}\end{equation}
     We can thus see that $S$ includes an infinite number of product states,
     all of the form $\gamma \left(x\ket{\alpha_1}^\sA +
     z\ket{\alpha_2}^\sA\right)
     \left(x\ket{\beta_1}^\sB + z \frac{ad}{bc}\ket{\beta_2}^\sB\right)$
     for some $x, z \in \mathbb{C}$ and a normalisation factor $\gamma$.
     Accordingly, we can define an invertible linear map
     $L:A_\tsep \mapsto B_\tsep$ mapping:
     \begin{eqnarray}
     L \ket{\alpha_1}^\sA &=& \ket{\beta_1}^\sB, \\
     L \ket{\alpha_2}^\sA &=& \frac{ad}{bc} \ket{\beta_2}^\sB,
     \end{eqnarray}
     and we conclude that the set of separable states in $S$ is exactly:
     \begin{equation}\label{eq:case22_sep}
     \left\{\ket{\Psi}^{\sA\sB}
     \triangleq \gamma \ket{\psi}^\sA \left(L \ket{\psi}\right)^\sB
     ~ : ~ \ket{\psi}^\sA \in A_\tsep ~ , ~ \gamma \in \mathbb{C} ~ , ~
     \left\|\ket{\Psi}^{\sA\sB}\right\| = 1 \right\}.
     \end{equation}
     Thus, the separable pure and mixed states over $S$
     are exactly all mixtures (convex combinations) of the states
     in Eq.~\eqref{eq:case22_sep}. Formally:
     \begin{equation}
     \cD^S_\tsep = \operatorname{conv}
     \left\{\proj{\Psi}^{\sA\sB}~:~\ket{\Psi}^{\sA\sB}
     = \gamma \ket{\psi}^\sA \left(L \ket{\psi}\right)^\sB~,~
     \ket{\psi}^\sA \in A_\tsep~,~\gamma \in \mathbb{C}~,~
     \left\|\ket{\Psi}^{\sA\sB}\right\| = 1 \right\},
     \end{equation}
     as needed.
\end{enumerate}
\vspace{-6.5mm}
\end{proof}

\section{Geometric descriptions and figures}
\label{sec:illustrations}
In Theorem~\ref{thm:main} we gave an algebraic description of the possible sets
$\cD^S_\tsep$ of separable pure and mixed states over $S$.
Now we can also give a geometric description of the sets $\cD^S_\tsep$
as closed convex subsets of the qutrit space $\cD(S)$
which is shown schematically in Fig.~\ref{fig:qutrit}.
The possible sets $\cD^S_\tsep$ are denoted according to the relevant case
in Theorem~\ref{thm:main}---that is, $\cD^{S_{(1,3)}}_\tsep$,
$\cD^{S_{(3,3)}}_\tsep$, $\cD^{S_{(2,3)-i}}_\tsep$, etc.

In this section we develop three-dimensional stereographic illustrations of
the possible classes from Theorem~\ref{thm:main}; because the sets $\cD^S_\tsep$
typically have dimension higher than $3$, we use lower-dimensional
sections and projections to present these schematic illustrations.
Fig.~\ref{fig:qutrit} is courtesy of Tadeusz Dorozi\`{n}ski~\cite{Dorozinski},
and the remaining figures (Figs.~\ref{fig:ctrit}--\ref{fig:phi-x-phi})
were produced using \emph{CalcPlot3D}~\cite{CalcPlot3D}.

Recall that $\cD(S)$ is the set of all density matrices
acting on the Hilbert space $S$,
and $\operatorname{conv}(A)$ is the convex hull of $A$---that is,
the set of all convex combinations (probabilistic mixtures)
of the states in the set $A \subseteq \cD(S)$.

We now present all possible classes from Theorem~\ref{thm:main}
and their geometric descriptions:

\begin{figure}[ht]
  \centering
  \includegraphics[width=6cm]{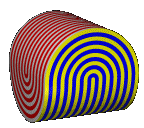} \qquad
  \includegraphics[width=5.5cm]{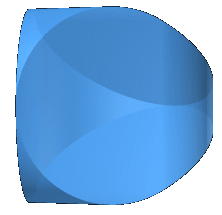}
  \caption{\textbf{The qutrit Hilbert space and cases (1,3) and (3,1):}
    $\cD(S) \equiv \cD(\CC^3)$ is an $8$-dimensional convex set
    whose extreme points
    form a $4$-dimensional smooth manifold, and whose other faces all
    have dimension $3$. It is a highly symmetric body with
    an $SU(3)$ symmetry group that acts transitively
    on the extreme points and also on the $3$-dimensional faces.
    It is also its own polar.
    There is no three-dimensional convex body presenting all these features,
    but other researchers have tried to construct
    visualisations~\cite{Bloor:qutrit,BWZ:apophatic,KKLM16,SWZ18,AK20,Qutrit-shape}
    that preserve some of the features,
    typically suggesting projections into three dimensions.
    Among these, the so-called sphericyl (left) and elliptic sphericyl (right)
    were previously considered in~\cite{BWZ:apophatic}
    as decent intuitive illustrations:
    they are convex hulls of smooth curves
    that are each a union of four tangent circular arcs
    on a sphere.
    These pictures are courtesy of Tadeusz Dorozi\`{n}ski~\cite{Dorozinski}.
  \label{fig:qutrit}}
\end{figure}

\medskip
\textbf{(1,3)/(3,1):} The set of separable states
$\cD^{S_{(1,3)}}_\tsep = \cD^{S_{(3,1)}}_\tsep = \cD(S)$
is the entire set of pure and mixed states over
the qutrit Hilbert space $S \equiv \CC^3$.
Unlike a qubit Hilbert space (whose corresponding Bloch sphere,
which represents all pure and mixed states over the qubit,
is three-dimensional and can thus be visualised),
the set of all pure and mixed states over a qutrit space is $8$-dimensional
and cannot be easily represented in three dimensions. Formally:
\begin{equation}
  \cD^{S_{(1,3)}}_\tsep = \cD^{S_{(3,1)}}_\tsep = \cD(S) \equiv \cD(\CC^3),
\end{equation}
the dimension of the set of separable states is indeed
$\dim \cD^{S_{(1,3)}}_\tsep = \dim \cD^{S_{(3,1)}}_\tsep = 8$,
and the set $S$ is the product of a fixed state with a local qutrit space:
in case (1,3), $S = \left\{\ket{\alpha_1}^\sA\right\} \ox
\Span\left\{\ket{\beta_1}^\sB, \ket{\beta_2}^\sB, \ket{\beta_3}^\sB\right\}$, and
in case (3,1), $S = \Span\left\{\ket{\alpha_1}^\sA, \ket{\alpha_2}^\sA,
\ket{\alpha_3}^\sA\right\} \ox \left\{\ket{\beta_1}^\sB\right\}$,
which are all product states.

A decent schematic illustration of the qutrit space can be obtained
as the convex hull of the seam on a tennis ball. For concreteness, we choose 
the sphericyl and elliptic sphericyl as described in Fig.~\ref{fig:qutrit}. 
Note that these are only intuitive ``cartoons'':
they only aim to demonstrate some of the qualitative features of $\cD(S)$
in three dimensions. In the subsequent cases we follow a more precise approach,
constructing three-dimensional sections/projections
of the resulting subsets $\cD^S_\tsep$.

\medskip
\textbf{(3,3) and (2,3)/(3,2)-i:} The set of separable states
$\cD^{S_{(3,3)}}_\tsep = \cD^{S_{(2,3)-i}}_\tsep = \cD^{S_{(3,2)-i}}_\tsep$
simply consists of all mixtures of the three product states
$\ket{\alpha_1}^\sA \ket{\beta_1}^\sB, \ket{\alpha_2}^\sA \ket{\beta_2}^\sB,
\ket{\alpha_3}^\sA \ket{\beta_3}^\sB$.
Formally:
\begin{equation}\begin{split}
\cD^{S_{(3,3)}}_\tsep 
  &= \cD^{S_{(2,3)-i}}_\tsep = \cD^{S_{(3,2)-i}}_\tsep \\
  &= \operatorname{conv}
       \left\{\proj{\alpha_1}^\sA \ox \proj{\beta_1}^\sB ~ , ~
              \proj{\alpha_2}^\sA \ox \proj{\beta_2}^\sB ~ , ~
              \proj{\alpha_3}^\sA \ox \proj{\beta_3}^\sB
       \right\},
\end{split}\end{equation}
and the dimension is $\dim \cD^{S_{(3,3)}}_\tsep = \dim \cD^{S_{(2,3)-i}}_\tsep
= \dim \cD^{S_{(3,2)-i}}_\tsep = 2$.
Geometrically, this set simply
forms a triangle with these three product states as corners
(which is exactly a classical two-dimensional probability simplex
of ternary probability distributions), as shown in Fig.~\ref{fig:ctrit}. 

\begin{figure}[ht]
  \centering
  \includegraphics[width=8cm]{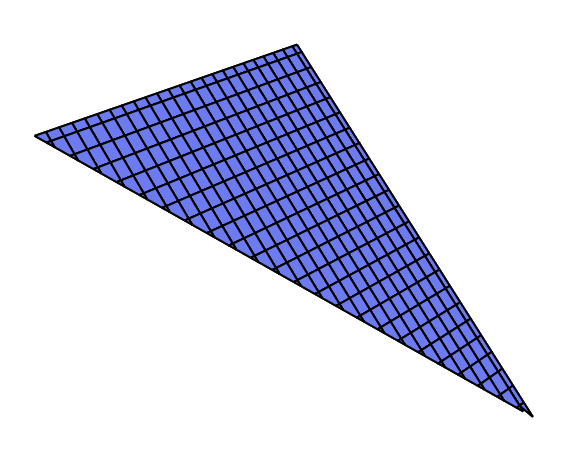}
  \caption{\textbf{Cases (3,3) and (2,3)/(3,2)-i:}
    The set of separable states $\cD^{S_{(3,3)}}_\tsep
    = \cD^{S_{(2,3)-i}}_\tsep = \cD^{S_{(3,2)-i}}_\tsep$ is
    a triangle whose vertices correspond to the pure product states
    $\proj{\alpha_1}^\sA \ox \proj{\beta_1}^\sB$,
    $\proj{\alpha_2}^\sA \ox \proj{\beta_2}^\sB$, and
    $\proj{\alpha_3}^\sA \ox \proj{\beta_3}^\sB$.
  \label{fig:ctrit}}
\end{figure}

\medskip
\textbf{(2,3)/(3,2)-ii:} In case (2,3)-ii, the set of separable states
$\cD^{S_{(2,3)-ii}}_\tsep$ consists of all mixtures of the state
$\proj{\alpha_1}^\sA \ox \proj{\beta_1}^\sB$
with any pure or mixed state over the local qubit space
$\left\{\ket{\alpha_2}^\sA\right\}
\ox \Span\left\{\ket{\beta_2}^\sB, \ket{\beta_3}^\sB\right\}$;
and similarly, in case (3,2)-ii, the set of separable states
$\cD^{S_{(3,2)-ii}}_\tsep$ consists of all mixtures of the state
$\proj{\alpha_1}^\sA \ox \proj{\beta_1}^\sB$
with any pure or mixed state over the local qubit space
$\Span\left\{\ket{\alpha_2}^\sA, \ket{\alpha_3}^\sA\right\}
\ox \left\{\ket{\beta_2}^\sB\right\}$.
Formally:
\begin{eqnarray}
\cD^{S_{(2,3)-ii}}_\tsep &=& \operatorname{conv}
\left[\left\{\proj{\alpha_1}^\sA \ox \proj{\beta_1}^\sB\right\} \cup
\cD\left(\left\{\ket{\alpha_2}^\sA\right\}
\ox \Span\left\{\ket{\beta_2}^\sB, \ket{\beta_3}^\sB\right\}\right)\right], \\
\cD^{S_{(3,2)-ii}}_\tsep &=& \operatorname{conv}
\left[\left\{\proj{\alpha_1}^\sA \ox \proj{\beta_1}^\sB\right\} \cup
\cD\left(\Span\left\{\ket{\alpha_2}^\sA, \ket{\alpha_3}^\sA\right\}
\ox \left\{\ket{\beta_2}^\sB\right\}\right)\right],
\end{eqnarray}
and the dimension is $\dim \cD^{S_{(2,3)-ii}}_\tsep
= \dim \cD^{S_{(3,2)-ii}}_\tsep = 4$.
Thus, in both cases, the set of separable states forms a spherical cone
connecting a single point (representing the state
$\proj{\alpha_1}^\sA \ox \proj{\beta_1}^\sB$)
with a Bloch sphere (representing the local qubit space)
in an overall four-dimensional space.
The extreme points of this set form two connected components: 
the single point and the two-dimensional surface of the Bloch sphere.
This body has a two-dimensional family of faces of dimension 1
and a single face of dimension three.
To illustrate this in a three-dimensional figure,
we present the convex hull (that is, the set of all convex combinations)
of the equatorial qubits of the Bloch sphere (a full disc)
with the single point, and we get a three-dimensional cone over a circular disc,
shown in Fig.~\ref{fig:qubit-cone}. 

\begin{figure}[ht]
  \centering
  \includegraphics[width=8cm]{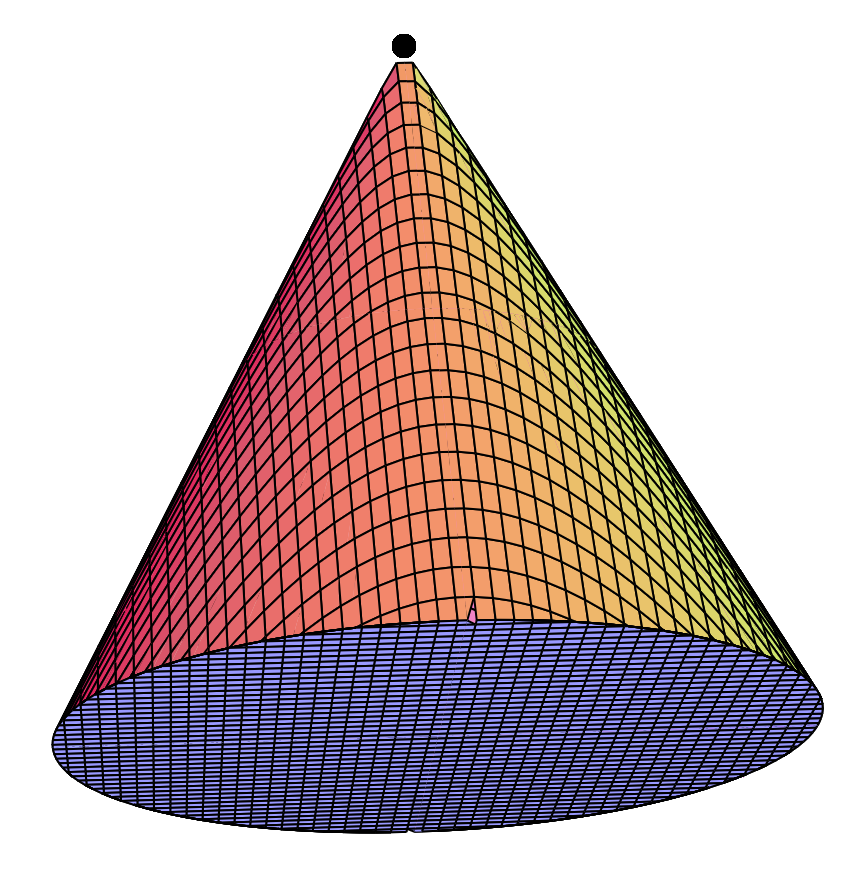}
  \caption{\textbf{Case (2,3)/(3,2)-ii:}
    The sets of separable states $\cD^{S_{(2,3)-ii}}_\tsep$
    and $\cD^{S_{(3,2)-ii}}_\tsep$ are $4$-dimensional spherical cones.
    Here we present one of their three-dimensional sections,
    which is a convex hull
    of the qubit states along the equator of the Bloch sphere
    in addition to a single point representing
    $\proj{\alpha_1}^\sA \ox \proj{\beta_1}^\sB$.
  \label{fig:qubit-cone}}
\end{figure}

\medskip
\textbf{(2,2)-i:} The set of separable states
$\cD^{S_{(2,2)-i}}_\tsep$ consists of all mixtures
of any pure or mixed state over the local qubit space
$S_1 \triangleq \left\{\ket{\alpha_2}^\sA\right\} \ox
\Span\left\{\ket{\beta_1}^\sB, \ket{\beta_2}^\sB\right\}$
with any pure or mixed state over the local qubit space
$S_2 \triangleq \Span\left\{\ket{\alpha_1}^\sA, \ket{\alpha_2}^\sA\right\} \ox
\left\{\ket{\beta_1}^\sB\right\}$; note that the intersection of the two spaces
is $S_1 \cap S_2 = \left\{\ket{\alpha_2}^\sA \ox \ket{\beta_1}^\sB\right\}$.
Formally:
\begin{equation}\begin{split}
\cD^{S_{(2,2)-i}}_\tsep &= \operatorname{conv}\left[\cD(S_1) \cup \cD(S_2)\right]
\\
&= \operatorname{conv}\left[\cD\left(\left\{\ket{\alpha_2}^\sA\right\} \ox
\Span\left\{\ket{\beta_1}^\sB, \ket{\beta_2}^\sB\right\}\right) \cup
\cD\left(\Span\left\{\ket{\alpha_1}^\sA, \ket{\alpha_2}^\sA\right\} \ox
\left\{\ket{\beta_1}^\sB\right\}\right)\right] \\
&= \operatorname{conv}\left[\cD\left(\left\{\ket{\alpha_2}^\sA\right\}
\ox B_\tsep\right)
\cup \cD\left(A_\tsep \ox \left\{\ket{\beta_1}^\sB\right\}\right)\right],
\end{split}\end{equation}
and the dimension is $\dim \cD^{S_{(2,2)-i}}_\tsep = 6$.
Thus, the set of separable states is the convex hull
(that is, the set of all convex combinations)
of two Bloch spheres intersecting at a single point,
where the single point represents $\ket{\alpha_2}^\sA \ox \ket{\beta_1}^\sB$.
The two Bloch spheres are transversal to each other:
they are two three-dimensional spaces intersecting at a single point.
We can thus repeat the same idea as in case (2,3)/(3,2)-ii
and take the equatorial qubits of both Bloch spheres,
which are two full discs transversal to each other;
however, the convex hull of the two discs is still $4$-dimensional.
We thus project this figure into three dimensions,
obtaining the convex hull of two circles intersecting at one point,
where their respective two-dimensional planes intersect at right angles,
as illustrated in Fig.~\ref{fig:qubit+qubit}. 

\begin{figure}[ht]
  \centering
  \includegraphics[width=8cm]{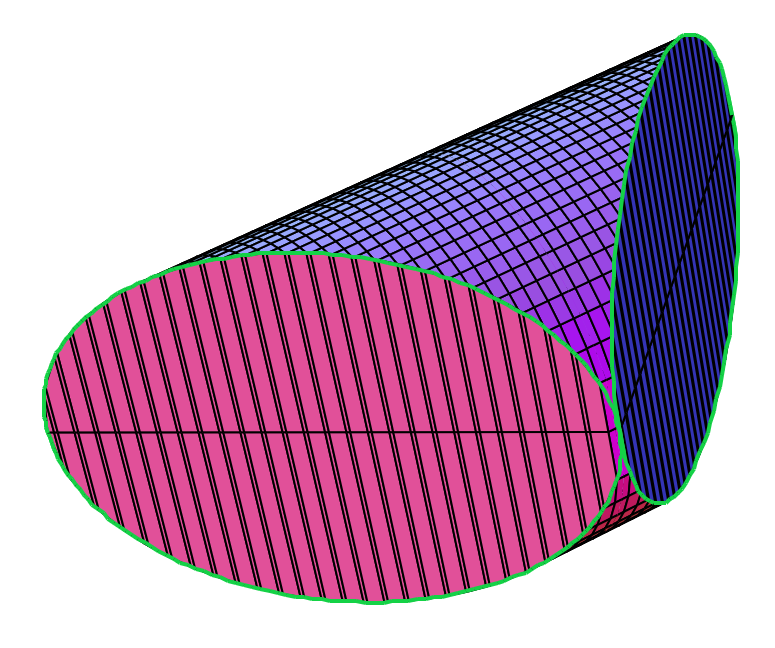}
  \caption{\textbf{Case (2,2)-i:}
    The set of separable states $\cD^{S_{(2,2)-i}}_\tsep$
    is the $6$-dimensional convex hull of two Bloch spheres intersecting
    at one point. Here we present one of its four-dimensional sections,
    which is a convex hull of two Bloch sphere
    equators intersecting at one point, projected into three dimensions.
  \label{fig:qubit+qubit}}
\end{figure}

\begin{figure}[ht]
  \centering
  \includegraphics[width=8cm]{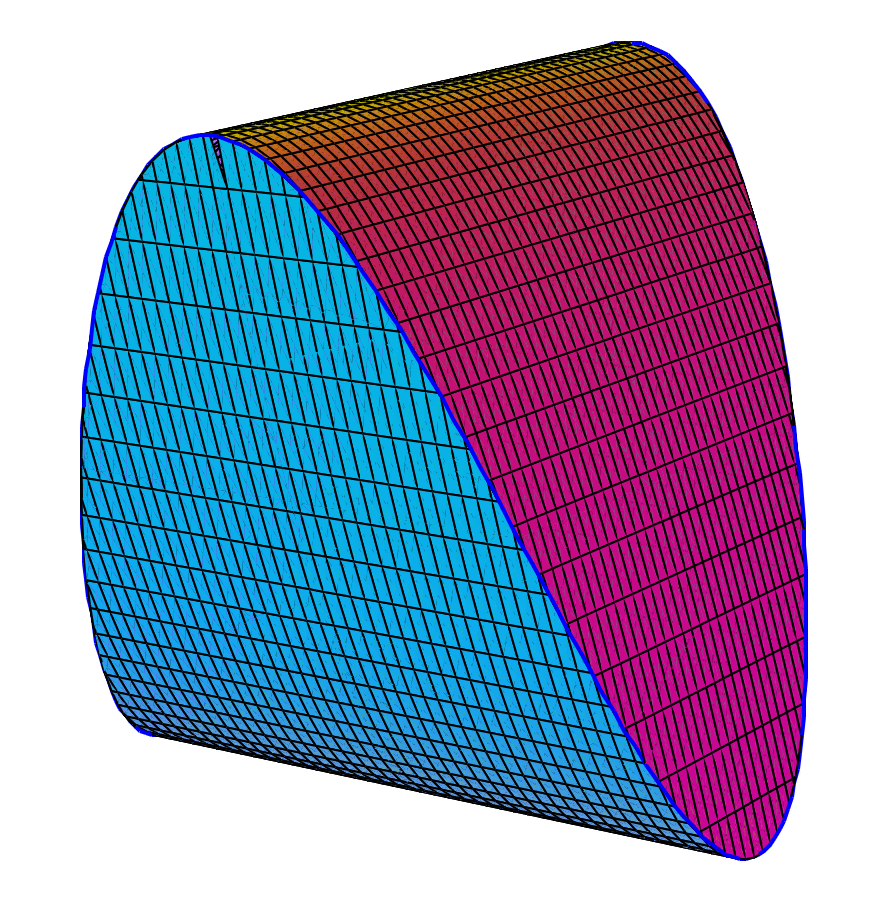}
  \caption{\textbf{Case (2,2)-ii:}
    This three-dimensional section of the full eight-dimensional set
    $\cD^{S_{(2,2)-i}}_\tsep$ is the convex hull of
    $\left(x=\cos \varphi ~ , ~ y=\sin \varphi ~ , ~
    z=\cos^2 \varphi-\sin^2 \varphi\right)$ for $0 \le \varphi \le 2\pi$,
    which is the intersection of two isomorphic degenerate paraboloids
    and represents a projection of
    the possible product states $\ket{\psi}^\sA \ket{\psi}^\sB$
    (up to local invertible maps), showing only those of
    the pure states $\ket{\psi}$ which are on the Bloch sphere's equator.
    Except for its one-dimensional manifold of extreme points,
    its only other faces are two disjoint one-dimensional families of lines.
  \label{fig:phi-x-phi}}
\end{figure}

\medskip
\textbf{(2,2)-ii:} The set of separable states
$\cD^{S_{(2,2)-ii}}_\tsep$ consists of all mixtures
of states of the form $\ket{\psi}^\sA \left(L \ket{\psi}\right)^\sB$,
up to normalisation, for a specific (known) invertible linear map $L$. Formally:
\begin{equation}
\cD^S_\tsep = \operatorname{conv}
\left\{\proj{\Psi}^{\sA\sB}~:~\ket{\Psi}^{\sA\sB}
= \gamma \ket{\psi}^\sA \left(L \ket{\psi}\right)^\sB~,~
\ket{\psi}^\sA \in \Span\left\{\ket{\alpha_1}^\sA, \ket{\alpha_2}^\sA\right\}~,~
\gamma \in \mathbb{C}~,~\left\|\ket{\Psi}^{\sA\sB}\right\| = 1 \right\},
\end{equation}
where, given that
$\ket{\alpha_3}^\sA = a\ket{\alpha_1}^\sA + b\ket{\alpha_2}^\sA$
and $\ket{\beta_3}^\sB = c\ket{\beta_1}^\sB + d\ket{\beta_2}^\sB$
for $a,b,c,d \in \CC \setminus \{0\}$ that were defined
in the proof of Theorem~\ref{thm:main} (case (2,2))
in Subsection~\ref{subsec:proof_22}, the map $L$ is defined as:
\begin{eqnarray}
L \ket{\alpha_1}^\sA &=& \ket{\beta_1}^\sB, \\
L \ket{\alpha_2}^\sA &=& \frac{ad}{bc} \ket{\beta_2}^\sB, \\
L \ket{\alpha_3}^\sA &=& a \ket{\beta_1}^\sB + \frac{ad}{c} \ket{\beta_2}^\sB
= \frac{a}{c} \ket{\beta_3}^\sB,
\end{eqnarray}
and the dimension is $\dim \cD^{S_{(2,2)-ii}}_\tsep = 8$.
Thus, the set of separable states is an $8$-dimensional convex subset
of the full qutrit space $\cD(S) \equiv \cD(\CC^3)$,
and its extreme points
form the two-dimensional smooth manifolds of pure product states
$\proj{\psi}^\sA \ox \left(L \proj{\psi} L^\dagger\right)^\sB$.
We notice that this is the only case (except the trivial case (1,3)/(3,1))
where the set of separable states has the same dimension
as the full space $\cD(S)$; in all other cases the set of separable states
has a lower dimension and is thus of measure zero.

An interesting special case is $ad=bc$,
because it satisfies $\left(L \ket{\psi}\right)^\sB \equiv \ket{\psi}^\sA$
if we apply the local invertible operations mapping
$\ket{\alpha_1}^\sA \mapsto \ket{1}^\sA$ ~ , ~
$\ket{\alpha_2}^\sA \mapsto \ket{2}^\sA$ ~ , ~
$\ket{\beta_1}^\sB \mapsto \ket{1}^\sB$, and
$\ket{\beta_2}^\sB \mapsto \ket{2}^\sB$
and take the equivalences
$\ket{1}^\sA \equiv \ket{1}^\sB$ and $\ket{2}^\sA \equiv \ket{2}^\sB$.
In this case, the set of product states in $S$ is equivalent to
the set of all states of form $\ket{\psi}^\sA \ket{\psi}^\sB$,
which is the set of all product states in the \emph{symmetric subspace}
(also known as the ``triplet subspace'')
of the two-qubit Hilbert space $\cH_2 \times \cH_2$ (that is,
the three-dimensional subspace spanned by
$\left\{\ket{1}^\sA \ket{1}^\sB ~ , ~ \ket{2}^\sA \ket{2}^\sB ~ , ~
\ket{\Psi^+}^{\sA\sB} \triangleq \frac{\ket{1}^\sA \ket{2}^\sB +
\ket{2}^\sA \ket{1}^\sB}{\sqrt{2}}\right\}$).
Moreover, we can see that given this special case,
the general case (2,2)-ii can simply be obtained by applying
a local invertible operation $X \ox Y$ (not affecting entanglement)
to the special case.

To obtain the figure, we look at the special case $ad = bc$;
we apply the local invertible operations mapping
$\ket{\alpha_1}^\sA \mapsto \ket{1}^\sA$ ~ , ~
$\ket{\alpha_2}^\sA \mapsto \ket{2}^\sA$ ~ , ~
$\ket{\beta_1}^\sB \mapsto \ket{1}^\sB$, and
$\ket{\beta_2}^\sB \mapsto \ket{2}^\sB$,
so that the product states are equivalent to $\ket{\psi}^\sA \ket{\psi}^\sB$
where $\ket{\psi}^\sA \in \Span \left\{\ket{1}^\sA, \ket{2}^\sA\right\}$;
and we limit our view to the section of equatorial qubit states---that is,
qubits of the form:
\begin{equation}
\ket{\psi}^\sA = \frac{\ket{1}^\sA + e^{i\varphi} \ket{2}^\sA}{\sqrt{2}}
\end{equation}
for $0 \le \varphi \le 2\pi$. All product states are thus of the form:
\begin{equation}
\begin{split}
\ket{\psi}^\sA \ket{\psi}^\sB
&= \frac{\ket{1}^\sA + e^{i\varphi} \ket{2}^\sA}{\sqrt{2}}
\otimes \frac{\ket{1}^\sB + e^{i\varphi} \ket{2}^\sB}{\sqrt{2}} \\
&= \frac{\ket{1}^\sA \ket{1}^\sB
+ e^{i\varphi} \ket{1}^\sA \ket{2}^\sB + e^{i\varphi} \ket{2}^\sA \ket{1}^\sB
+ e^{2i\varphi} \ket{2}^\sA \ket{2}^\sB}{2} \\
&= \frac{\ket{11}^{\sA\sB}}{2}
+ \frac{e^{i\varphi} \ket{\Psi^+}^{\sA\sB}}{\sqrt{2}}
+ \frac{e^{2i\varphi} \ket{22}^{\sA\sB}}{2},
\end{split}
\end{equation}
(where we denote $\ket{\Psi^+}^{\sA\sB}
\triangleq \frac{\ket{1}^\sA \ket{2}^\sB + \ket{2}^\sA \ket{1}^\sB}{\sqrt{2}}$),
and the resulting density matrices are:
\begin{equation}
\begin{split}
\proj{\psi}^\sA \otimes \proj{\psi}^\sB
&= \frac{\proj{11}^{\sA\sB}}{4} + \frac{\proj{\Psi^+}^{\sA\sB}}{2}
+ \frac{\proj{22}^{\sA\sB}}{4} \\
&+ e^{i\varphi} \frac{\ket{\Psi^+}\bra{11}^{\sA\sB}}{2\sqrt{2}}
+ e^{-i\varphi} \frac{\ket{11}\bra{\Psi^+}^{\sA\sB}}{2\sqrt{2}}
+ e^{-i\varphi} \frac{\ket{\Psi^+}\bra{22}^{\sA\sB}}{2\sqrt{2}} \\
&+ e^{i\varphi} \frac{\ket{22}\bra{\Psi^+}^{\sA\sB}}{2\sqrt{2}}
+ e^{2i\varphi} \frac{\ket{22}\bra{11}^{\sA\sB}}{4}
+ e^{-2i\varphi} \frac{\ket{11}\bra{22}^{\sA\sB}}{4}.
\end{split}
\end{equation}
Thus, the resulting density matrices are of the following form:
\begin{equation}
\proj{\psi}^\sA \otimes \proj{\psi}^\sB
= C + e^{i\varphi} A + e^{-i\varphi} A^\dagger
+ e^{2i\varphi} B + e^{-2i\varphi} B^\dagger
= C + e^{i\varphi} A + \left(e^{i\varphi} A\right)^\dagger
+ e^{2i\varphi} B + \left(e^{2i\varphi} B\right)^\dagger,
\end{equation}
where $A, A^\dagger, B, B^\dagger, C$ are linearly independent matrices.
The resulting set of separable states is the convex hull
of an infinite number of extreme points (representing these product states)
forming a smooth curve;
these extreme points are affine-linearly parametrised by the pair of complex numbers
$\left(e^{i\varphi}, e^{2i\varphi}\right)$ for $0 \le \varphi \le 2\pi$.
Since the resulting body is still four-dimensional,
we project it into three dimensions by retaining only the real part
of the second complex number $e^{2i\varphi}$, getting the parametrisation
$\left(e^{i\varphi} ~ , ~
\cos(2\varphi) = \cos^2 \varphi - \sin^2 \varphi\right)$, or equivalently
$\left(\cos \varphi ~ , ~ \sin \varphi ~ , ~
\cos^2 \varphi - \sin^2 \varphi\right)$.
Using real number parameters $x,y,z \in \RR$, the obtained convex body is
\begin{equation}
  \operatorname{conv}\left\{(x,y,z) ~ : ~ x^2+y^2=1 ~ , ~ z=x^2-y^2\right\}
   = \left\{(x,y,z) ~ : ~ x^2+y^2\leq 1 ~ , ~ 2x^2-1\leq z \leq 1-2y^2\right\},
\end{equation}
which is shown in Fig.~\ref{fig:phi-x-phi}.

\section{Generalisation to multipartite systems}
\label{sec:multipartite}
We can now generalise our result regarding bipartite systems
(Theorem~\ref{thm:main}) to all multipartite systems.
In~\cite{BLM}, the generalised classification
of \emph{fully separable} states in two-dimensional subspaces
of multipartite systems (Theorem~11 in~\cite{BLM})
turned out to be completely identical to the classification of
separable states in two-dimensional subspaces of bipartite systems
(Theorem~9 in~\cite{BLM}, stated in our paper as Theorem~\ref{thm:BLM}).
However, in our paper, we will be able to find several differences
between the bipartite case and the multipartite case:
most notably, some bipartite classes disappear for multipartite systems.

Similarly to Section~\ref{sec:setting},
we begin our analysis with a rank-3 quantum mixed state
$\rho \triangleq \rho^{\sA_1 \cdots \sA_k}$ on a $k$-partite system
$\cH_{\sA_1} \ox \cH_{\sA_2} \ox \cdots \ox \cH_{\sA_k}$
(where $k \geq 3$).
As before, the support of $\rho$ is the three-dimensional
subspace $S = \supp \rho$ which is spanned by the eigenstates of $\rho$,
and we define:
\begin{equation}\begin{split}
S_\tsep &\triangleq \Span \left\{\ket{\psi} \in S ~ : ~
\ket{\psi}\text{ is a product state}\right\} \\
&= \Span \left\{\ket{\psi} \in S ~ : ~
\exists \ket{\phi^{(1)}} \in \cH_{\sA_1} ~ , ~ \ldots ~ , ~
\ket{\phi^{(k)}} \in \cH_{\sA_k} ~ : ~
\ket{\psi} = \ket{\phi^{(1)}} \otimes \cdots \otimes \ket{\phi^{(k)}}\right\}.
\end{split}\end{equation}
Since the cases $\dim S_\tsep \in \{0,1,2\}$ were analysed in~\cite{BLM},
we focus again on the case $\dim S_\tsep = 3$, where $S_\tsep = S$
and it is spanned by three linearly independent product states:
\begin{equation}\label{eq:S_multi}
S = \Span \left\{
\ket{\alpha_1^{(1)}} \otimes \cdots \otimes \ket{\alpha_1^{(k)}} ~ , ~
\ket{\alpha_2^{(1)}} \otimes \cdots \otimes \ket{\alpha_2^{(k)}} ~ , ~
\ket{\alpha_3^{(1)}} \otimes \cdots \otimes \ket{\alpha_3^{(k)}}\right\}.
\end{equation}
We can now define the set of fully separable states we are going to analyse:
\begin{equation}
\cD^S_\tsep \triangleq \{\rho \in \cD(S) ~ : ~ \rho\text{ is fully separable}\},
\end{equation}
where $\cD(S)$ is (as before) the set of all density matrices
over the Hilbert space $S$, and a state $\rho$ is said to be ``fully separable''
if it is a mixture of product states
$\ket{\phi^{(1)}} \otimes \cdots \otimes \ket{\phi^{(k)}}$.

Similarly to Theorem~\ref{thm:main}, our analysis will depend on
the dimensions of the $k$ local subspaces $A^{(j)}_\tsep$ (for $1 \le j \le k$),
which are defined as follows:
\begin{equation}\label{eq:Ajsep_multi}
A^{(j)}_\tsep \triangleq \Span \left\{\ket{\alpha_1^{(j)}}, \ket{\alpha_2^{(j)}},
\ket{\alpha_3^{(j)}}\right\},
\end{equation}
and each of them can be $1$-, $2$-, or $3$-dimensional.

We would now like to point out that whenever the dimension of $A^{(j)}_\tsep$
is $1$, it does not contribute any genuine entanglement and can be removed
from our analysis. Indeed, if $\dim A^{(j)}_\tsep = 1$, then necessarily
$\ket{\alpha_1^{(j)}} = \ket{\alpha_2^{(j)}} = \ket{\alpha_3^{(j)}}$,
in which case all states in $S$ are actually tensor products of
a $(k-1)$-partite state with a constant state $\ket{\alpha_1^{(j)}}$.
This state can be trivially absorbed into another system without affecting
the analysis of entanglement and separability.
Thus, the state is \emph{genuinely} only $(k-1)$-partite, not $k$-partite.
(A similar observation regarding the bipartite case of Theorem~\ref{thm:main}
is that systems corresponding to its cases (1,3) and (3,1)
are actually not bipartite at all, but ``monopartite''.)
We can therefore assume, without loss of generality, that
$\dim A^{(j)}_\tsep \in \{2,3\}$ for all $1 \le j \le k$,
thus focusing on \emph{genuinely} multipartite systems with $k \ge 3$ subsystems.

For the analysis, we shall also use the following Lemma,
which (informally) says that three quantum states that are either
\emph{linearly independent} or \emph{linearly dependent but in general position}
become linearly independent when we take their tensor product
with another (non-trivial) Hilbert space:
\begin{lemma}\label{lemma:gp}
In a bipartite Hilbert space $\cH_\sA \ox \cH_\sB$,
if three quantum states $\ket{\alpha_1}^\sA, \ket{\alpha_2}^\sA,
\ket{\alpha_3}^\sA \in \cH_\sA$ are either \emph{linearly independent}
or \emph{linearly dependent but in general position}
(that is, any two of them are linearly independent),
then for any three quantum states $\ket{\beta_1}^\sB,
\ket{\beta_2}^\sB, \ket{\beta_3}^\sB \in \cH_\sB$
that are not all identical (up to normalisation and global phase),
the three states $\ket{\alpha_1}^\sA \ket{\beta_1}^\sB ~ , ~
\ket{\alpha_2}^\sA \ket{\beta_2}^\sB ~ , ~
\ket{\alpha_3}^\sA \ket{\beta_3}^\sB \in \cH_\sA \ox \cH_\sB$
are \emph{linearly independent}.
\end{lemma}
\begin{proof}
Assume by contradiction that the three states
$\ket{\alpha_1}^\sA \ket{\beta_1}^\sB ~ , ~
\ket{\alpha_2}^\sA \ket{\beta_2}^\sB ~ , ~
\ket{\alpha_3}^\sA \ket{\beta_3}^\sB$
are linearly dependent. They must be in general position,
because any two of them cannot be equal
(otherwise, two of $\ket{\alpha_1}^\sA, \ket{\alpha_2}^\sA,
\ket{\alpha_3}^\sA$ would be equal).
Therefore, without loss of generality,
$\ket{\alpha_3}^\sA \ket{\beta_3}^\sB \in
\Span \left\{\ket{\alpha_1}^\sA \ket{\beta_1}^\sB ~ , ~
\ket{\alpha_2}^\sA \ket{\beta_2}^\sB\right\}$,
so there are $a,b \in \CC \setminus \{0\}$ such that:
\begin{equation}\label{eq:proof_depen}
\ket{\alpha_3}^\sA \ket{\beta_3}^\sB =
a \ket{\alpha_1}^\sA \ket{\beta_1}^\sB +
b \ket{\alpha_2}^\sA \ket{\beta_2}^\sB.
\end{equation}
This equation means, in particular, that $\ket{\alpha_1}^\sA, \ket{\alpha_2}^\sA,
\ket{\alpha_3}^\sA \in \cH_\sA$ cannot be linearly independent,
so they must be linearly dependent and, by assumption, in general position;
this observation implies that $\ket{\alpha_3}^\sA \in \Span
\left\{\ket{\alpha_1}^\sA, \ket{\alpha_2}^\sA\right\}$,
so there are $x,y \in \CC \setminus \{0\}$ such that:
\begin{equation}
\ket{\alpha_3}^\sA = x \ket{\alpha_1}^\sA + y \ket{\alpha_2}^\sA,
\end{equation}
and substituting this in Eq.~\eqref{eq:proof_depen}, we find:
\begin{equation}
\left(x \ket{\alpha_1}^\sA + y \ket{\alpha_2}^\sA\right) \ket{\beta_3}^\sB =
a \ket{\alpha_1}^\sA \ket{\beta_1}^\sB +
b \ket{\alpha_2}^\sA \ket{\beta_2}^\sB,
\end{equation}
or equivalently,
\begin{equation}
\ket{\alpha_1}^\sA \left(x \ket{\beta_3}^\sB - a \ket{\beta_1}^\sB\right) =
\ket{\alpha_2}^\sA \left(b \ket{\beta_2}^\sB - y \ket{\beta_3}^\sB\right).
\end{equation}
Because we know that $\ket{\alpha_1}^\sA \ne \ket{\alpha_2}^\sA$,
this necessarily implies $x \ket{\beta_3}^\sB - a \ket{\beta_1}^\sB = 0$
and $b \ket{\beta_2}^\sB - y \ket{\beta_3}^\sB = 0$.
Because we know that $a,b,x,y \ne 0$, it holds that:
\begin{equation}
\frac{a}{x} \ket{\beta_1}^\sB = \ket{\beta_3}^\sB
= \frac{b}{y} \ket{\beta_2}^\sB,
\end{equation}
which contradicts the assumption that the three quantum states
$\ket{\beta_1}^\sB, \ket{\beta_2}^\sB, \ket{\beta_3}^\sB \in \cH_\sB$
are not all identical (up to normalisation and global phase).
Thus, the three bipartite states mentioned above,
$\ket{\alpha_1}^\sA \ket{\beta_1}^\sB ~ , ~
\ket{\alpha_2}^\sA \ket{\beta_2}^\sB ~ , ~
\ket{\alpha_3}^\sA \ket{\beta_3}^\sB$,
must be linearly independent, as we wanted.
\end{proof}

For the full analysis, we use the following notations:
the set of integer numbers from $1$ to $k$
is denoted by $[k]$ (formally, $[k] \triangleq \{1,2,\ldots,k\}$);
and for any subset $L \subseteq [k]$ and $1 \le i \le 3$, the product of all
states $\ket{\alpha_i^{(j)}}$ for all $j \in L$ is denoted $\ket{\alpha_i^{(L)}}$
(formally, $\ket{\alpha_i^{(L)}} \triangleq
\bigotimes_{j \in L} \ket{\alpha_i^{(j)}}$).
Similarly, for all $L \subseteq [k]$ we define the generalised local set
$A^{(L)}_\tsep$ as the span of all states $\ket{\alpha_i^{(L)}}$:
\begin{equation}\label{eq:ALsep_multi}
A^{(L)}_\tsep \triangleq \Span \left\{\ket{\alpha_1^{(L)}}, \ket{\alpha_2^{(L)}},
\ket{\alpha_3^{(L)}}\right\}.
\end{equation}

We can now present the Theorem which classifies
all three-dimensional subspaces of genuinely multipartite
systems ($k$-partite systems with $k \ge 3$, where none of the local subsystems
is one-dimensional) into two general cases,
corresponding to Fig.~\ref{fig:ctrit} (a triangle) and
Fig.~\ref{fig:qubit-cone} (a spherical cone):

\begin{theorem}
  \label{thm:multipartite}
  Given a multipartite Hilbert space
  $\cH_{\sA_1} \ox \cH_{\sA_2} \ox \cdots \ox \cH_{\sA_k}$
  and any three-dimensional subspace
  $S \subseteq \cH_{\sA_1} \ox \cH_{\sA_2} \ox \cdots \ox \cH_{\sA_k}$,
  let $S_\tsep$ be the subspace spanned by all product states in $S$
  (that is, $S_\tsep \triangleq \Span \{\ket{\psi} \in S ~ : ~
  \ket{\psi}\text{ is a product state}\}$).
  We would like to characterise the set of all fully separable states
  (both pure and mixed) over $S$, a set we denote by $\cD^S_\tsep$
  and formally define as follows:
  \begin{equation}
  \cD^S_\tsep \triangleq \{\rho \in \cD(S) ~ : ~ \rho\text{ is fully separable}\}.
  \end{equation}
  
  If $\dim S_\tsep \leq 2$, then $\cD^S_\tsep$
  belongs to one of the five classes described in Theorem~11 of~\cite{BLM}
  (or, equivalently, to one of the five multipartite generalisations
  of the classes described in Theorem~\ref{thm:BLM}).
  
  If $\dim S_\tsep = 3$ (so $S_\tsep=S$),
  then using the notations of Eqs.~\eqref{eq:S_multi}
  and~\eqref{eq:Ajsep_multi}:
  \begin{equation}\begin{split}
  S &= \Span \left\{
  \ket{\alpha_1^{(1)}} \otimes \cdots \otimes \ket{\alpha_1^{(k)}} ~ , ~
  \ket{\alpha_2^{(1)}} \otimes \cdots \otimes \ket{\alpha_2^{(k)}} ~ , ~
  \ket{\alpha_3^{(1)}} \otimes \cdots \otimes \ket{\alpha_3^{(k)}}\right\}, \\
  A^{(j)}_\tsep &\triangleq \Span \left\{\ket{\alpha_1^{(j)}},
  \ket{\alpha_2^{(j)}}, \ket{\alpha_3^{(j)}}\right\},
  \end{split}\end{equation}
  and assuming, without loss of generality (as explained above),
  that $\dim A^{(j)}_\tsep \in \{2,3\}$ for all $1 \le j \le k$,
  one of the two following cases occurs:

  \begin{enumerate}
    \item\label{multi_triangle} \textbf{Triangle:}
       The fully separable pure and mixed states over $S$
       are exactly all mixtures (convex combinations)
       of $\proj{\alpha_1^{(1)}} \ox \cdots \ox \proj{\alpha_1^{(k)}}$ ~ , ~
       $\proj{\alpha_2^{(1)}} \ox \cdots \ox \proj{\alpha_2^{(k)}}$,
       and $\proj{\alpha_3^{(1)}} \ox \cdots \ox \proj{\alpha_3^{(k)}}$;
       all the other pure and mixed states over $S$ are entangled.
       Formally:
       \begin{equation}\begin{split}
       \cD^S_\tsep = \operatorname{conv}
       &\left\{\proj{\alpha_1^{(1)}} \ox \cdots \ox \proj{\alpha_1^{(k)}} ~ , ~
       \proj{\alpha_2^{(1)}} \ox \cdots \ox \proj{\alpha_2^{(k)}} ~ , \right. \\
       &~~\left. \proj{\alpha_3^{(1)}} \ox \cdots \ox
       \proj{\alpha_3^{(k)}}\right\}.
       \end{split}\end{equation}
       This case is a generalisation of the bipartite cases (3,3), (2,3)-i,
       and (3,2)-i in Theorem~\ref{thm:main},
       and it is illustrated in Fig.~\ref{fig:ctrit}.

    \item\label{multi_cone} \textbf{Spherical cone:}
        There exists a subsystem $\mathcal{H}_{\sA_\ell}$
        (namely, there exists an index $1 \le \ell \le k$)
        such that out of the three states
        $\ket{\alpha_1^{([k] \setminus \{\ell\})}},
        \ket{\alpha_2^{([k] \setminus \{\ell\})}},
        \ket{\alpha_3^{([k] \setminus \{\ell\})}}$,
        two states are equal (without loss of generality,
        we can assume $\ket{\alpha_2^{([k] \setminus \{\ell\})}} =
        \ket{\alpha_3^{([k] \setminus \{\ell\})}}$),
        and the fully separable pure and mixed states over $S$
        are exactly all mixtures (convex combinations) of
        $\proj{\alpha_1^{(1)}} \ox \cdots \ox \proj{\alpha_1^{(k)}}$
        with any pure or mixed state over the space
        $\Span\left\{\ket{\alpha_2^{(\ell)}},
        \ket{\alpha_3^{(\ell)}}\right\}
        \ox \left\{\ket{\alpha_2^{([k]\setminus\{\ell\})}}\right\}$;
        all the other pure and mixed states over $S$ are entangled.
        Formally, in case $\ket{\alpha_2^{([k] \setminus \{\ell\})}} =
        \ket{\alpha_3^{([k] \setminus \{\ell\})}}$:
        \begin{equation}\begin{split}
        \cD^S_\tsep = &\operatorname{conv}
        \left[\left\{\proj{\alpha_1^{(1)}} \ox \cdots
        \ox \proj{\alpha_1^{(k)}}\right\} \right. \\
        &\left. ~~~~ \cup ~ \cD\left(\Span\left\{\ket{\alpha_2^{(\ell)}},
        \ket{\alpha_3^{(\ell)}}\right\} \ox
        \left\{\ket{\alpha_2^{([k]\setminus\{\ell\})}}\right\}\right)\right],
        \end{split}\end{equation}
        and symmetric results are obtained
        in case $\ket{\alpha_1^{([k] \setminus \{\ell\})}} =
        \ket{\alpha_2^{([k] \setminus \{\ell\})}}$
        or $\ket{\alpha_1^{([k] \setminus \{\ell\})}} =
        \ket{\alpha_3^{([k] \setminus \{\ell\})}}$.
        This case is a generalisation of the bipartite case (3,2)-ii
        (and its symmetric case (2,3)-ii) in Theorem~\ref{thm:main},
        and it is illustrated in Fig.~\ref{fig:qubit-cone}.
  \end{enumerate}
\end{theorem}
\begin{remark}
  We point out that there is \emph{no} generalisation of neither bipartite cases
  (2,2)-i and (2,2)-ii of Theorem~\ref{thm:main}.
  Thus, both cases (2,2)-i and (2,2)-ii
  (illustrated in Figs.~\ref{fig:qubit+qubit}
  and~\ref{fig:phi-x-phi}, respectively) exist in bipartite systems
  but do not exist in (genuinely) multipartite systems.
\end{remark}
\begin{proof}
If $\dim S_\tsep \leq 2$, then applying Theorem~11 of~\cite{BLM}
to a two-dimensional subspace of $S$ which includes $S_\tsep$ proves that
the set $\cD^{S_\tsep}_\tsep$ of pure and mixed separable states over $S_\tsep$
belongs to one of the five classes described in~\cite{BLM}
for multipartite systems.
All the other pure and mixed states over $S$ (that are not states over $S_\tsep$)
must be entangled (that is, they cannot be fully separable).
Therefore, $\cD^S_\tsep = \cD^{S_\tsep}_\tsep$ indeed
belongs to one of the five classes described in~\cite{BLM}
for two-dimensional subspaces of multipartite systems,
which are direct generalisations of the classes
described in~\cite{BLM} (or in our Theorem~\ref{thm:BLM})
for two-dimensional subspaces of bipartite systems.

If $\dim S_\tsep = 3$, then $S_\tsep = S$,
in which case the proof typically applies Theorem~\ref{thm:main} to $S$
under specific bipartite \emph{partitions} (or \emph{cuts})
and analyses entanglement and separability with respect to these partitions.
For our $k$-partite system
$\cH_{\sA_1} \ox \cH_{\sA_2} \ox \cdots \ox \cH_{\sA_k}$ and the set of indexes
$[k] \triangleq \{1, 2, \ldots, k\}$, a bipartite partition is a pair $(X,Y)$
of two disjoint sets of indexes $X,Y \subseteq [k]$ such that $X \cup Y = [k]$
and $X \cap Y = \emptyset$. A state is said to be \emph{entangled (or separable)
with respect to the partition $(X,Y)$} if merging most subsystems and leaving
only the two large quantum subsystems $\cH_X \otimes \cH_Y \triangleq
\left( \bigotimes_{j \in X} \cH_{\sA_j} \right) \ox
\left( \bigotimes_{j \in Y} \cH_{\sA_j} \right)$
(reordering the quantum subsystems) results in
an entangled (or separable) bipartite state.
In particular, if a state is fully separable, it is separable with respect
to all possible bipartite partitions;
therefore, to prove that a multipartite state is entangled
(i.e., not fully separable), it is sufficient to
find a bipartite partition such that the state is entangled with respect to it.

We can now analyse the entanglement class of $S$ in three different situations,
given the dimensions of the $k$ local subspaces $A^{(1)}_\tsep, A^{(2)}_\tsep,
\ldots, A^{(k)}_\tsep$ (all dimensions are in $\{2, 3\}$)
defined in Eq.~\eqref{eq:Ajsep_multi}:

\begin{description}
\itemsep0em
\item[Situation 1: There are at least \emph{two} local subspaces
      of dimension $\mathbf{3}$]
       Formally, in this case, there exist $\ell, m \in [k]$ ($\ell \ne m$)
       such that $\dim A^{(\ell)}_\tsep = \dim A^{(m)}_\tsep = 3$.
       We can therefore analyse the entanglement with respect to
       the bipartite partition $(\{\ell\} ~ , ~ [k] \setminus \{\ell\})$:
       the dimensions of the relevant local subspaces
       are $\dim A^{(\ell)}_\tsep = 3$
       and $\dim A^{([k] \setminus \{\ell\})}_\tsep \ge \dim A^{(m)}_\tsep = 3$
       (according to Lemma~\ref{lemma:gp}), so the relevant case is (3,3).

       Therefore, applying Theorem~\ref{thm:main} (case (3,3))
       to $S$ under this partition
       implies that we are in case~\ref{multi_triangle} (a triangle):
       the fully separable pure and mixed states over $S$
       are exactly all mixtures (convex combinations)
       of $\proj{\alpha_1^{(1)}} \ox \cdots \ox \proj{\alpha_1^{(k)}}$ ~ , ~
       $\proj{\alpha_2^{(1)}} \ox \cdots \ox \proj{\alpha_2^{(k)}}$,
       and $\proj{\alpha_3^{(1)}} \ox \cdots \ox \proj{\alpha_3^{(k)}}$,
       and all the other pure and mixed states over $S$ are entangled.

\item[Situation 2: There is only \emph{one} local subspace
      of dimension $\mathbf{3}$]
        Formally, in this case, there exists $\ell \in [k]$
        such that $\dim A^{(\ell)}_\tsep = 3$,
        and for any other $m \in [k] \setminus \{\ell\}$ it holds that
        $\dim A^{(m)}_\tsep = 2$. We divide into two cases:
        \begin{enumerate}
        \item There is a subsystem $m \in [k] \setminus \{\ell\}$ such that
        the three states $\ket{\alpha_1^{(m)}}, \ket{\alpha_2^{(m)}},
        \ket{\alpha_3^{(m)}}$ are linearly dependent but in general position.
        In this case, if we apply Lemma~\ref{lemma:gp}
        we get that the three states $\ket{\alpha_1^{([k] \setminus \{\ell\})}},
        \ket{\alpha_2^{([k] \setminus \{\ell\})}},
        \ket{\alpha_3^{([k] \setminus \{\ell\})}}$ are linearly independent,
        so $\dim A^{([k] \setminus \{\ell\})}_\tsep = 3$.

        We can therefore analyse the entanglement with respect to
        the bipartite partition $(\{\ell\} ~ , ~ [k] \setminus \{\ell\})$:
        the dimensions of the relevant local subspaces are
        $\dim A^{(\ell)}_\tsep = \dim A^{([k] \setminus \{\ell\})}_\tsep = 3$,
        so the relevant case is (3,3).
        Therefore, applying Theorem~\ref{thm:main} (case (3,3))
        to $S$ under this partition
        implies that we are in case~\ref{multi_triangle} (a triangle),
        identically to Situation~1 above.
        \item For all subsystems $m \in [k] \setminus \{\ell\}$,
        two of the three states $\ket{\alpha_1^{(m)}}, \ket{\alpha_2^{(m)}},
        \ket{\alpha_3^{(m)}}$ are equal to one another. Formally, for each
        $m \in [k] \setminus \{\ell\}$ there exists a two-element subset
        $I_m \triangleq \{i_m, j_m\} \subset \{1,2,3\}$ such that
        $\ket{\alpha_{i_m}^{(m)}} = \ket{\alpha_{j_m}^{(m)}}$.
        We divide into two subcases:
        \begin{enumerate}
        \item If there exist two subsets $I_m$ and $I_n$
        which are different from one another (formally, if there exist
        $m, n \in [k] \setminus \{\ell\}$ such that $I_m \ne I_n$),
        then the three states $\ket{\alpha_1^{(m)}} \ket{\alpha_1^{(n)}} ~ , ~
        \ket{\alpha_2^{(m)}} \ket{\alpha_2^{(n)}} ~ , ~
        \ket{\alpha_3^{(m)}} \ket{\alpha_3^{(n)}}$ must be linearly independent.
        Therefore, according to Lemma~\ref{lemma:gp},
        the three states $\ket{\alpha_1^{([k] \setminus \{\ell\})}},
        \ket{\alpha_2^{([k] \setminus \{\ell\})}},
        \ket{\alpha_3^{([k] \setminus \{\ell\})}}$
        are, too, linearly independent,
        so $\dim A^{([k] \setminus \{\ell\})}_\tsep = 3$.

        We can therefore analyse the entanglement with respect to
        the bipartite partition $(\{\ell\} ~ , ~ [k] \setminus \{\ell\})$:
        the dimensions of the relevant local subspaces are
        $\dim A^{(\ell)}_\tsep = \dim A^{([k] \setminus \{\ell\})}_\tsep = 3$,
        so the relevant case is (3,3).
        Therefore, applying Theorem~\ref{thm:main} (case (3,3))
        to $S$ under this partition
        implies that we are in case~\ref{multi_triangle} (a triangle),
        identically to the two cases above.
        \item If all subsets $I_n$ are identical to each other
        for all $n \in [k] \setminus \{\ell\}$,
        then we can assume, without loss of generality, $I_n = \{2,3\}$.
        Then, for all $n \in [k] \setminus \{\ell\}$ it holds that
        $\ket{\alpha_2^{(n)}} = \ket{\alpha_3^{(n)}}$, which in particular
        implies $\ket{\alpha_2^{([k] \setminus \{\ell\})}} =
        \ket{\alpha_3^{([k] \setminus \{\ell\})}}$.
        Therefore, $\dim A^{([k] \setminus \{\ell\})}_\tsep = 2$.

        We can therefore analyse the entanglement with respect to
        the bipartite partition $(\{\ell\} ~ , ~ [k] \setminus \{\ell\})$:
        the dimensions of the relevant local subspaces
        are $\dim A^{(\ell)}_\tsep = 3$
        and $\dim A^{([k] \setminus \{\ell\})}_\tsep = 2$
        (where $\ket{\alpha_2^{([k] \setminus \{\ell\})}} =
        \ket{\alpha_3^{([k] \setminus \{\ell\})}}$
        are \emph{not} in general position),
        so the relevant case is (3,2)-ii.

        Therefore, applying Theorem~\ref{thm:main} (case (3,2)-ii)
        to $S$ under this partition
        implies that we are in case~\ref{multi_cone} (a spherical cone):
        the fully separable pure and mixed states over $S$
        are exactly all mixtures (convex combinations) of
        $\proj{\alpha_1^{(1)}} \ox \cdots \ox \proj{\alpha_1^{(k)}}$
        with any pure or mixed state over the space
        $\Span\left\{\ket{\alpha_2^{(\ell)}}, \ket{\alpha_3^{(\ell)}}\right\}
        \ox \left\{\ket{\alpha_2^{([k] \setminus \{\ell\})}}\right\}$,
        and all the other pure and mixed states over $S$ are entangled.
        \end{enumerate}
        \end{enumerate}

\item[Situation 3: \emph{All} local subspaces are of dimension $\mathbf{2}$]
        Formally, in this case, for all $j \in [k]$
        it holds that $\dim A^{(j)}_\tsep = 2$. We divide into three cases:
        \begin{enumerate}
        \item There are \emph{two} subsystems $\ell, m \in [k]$ such that
        the three states $\ket{\alpha_1^{(\ell)}}, \ket{\alpha_2^{(\ell)}},
        \ket{\alpha_3^{(\ell)}}$ are linearly dependent but in general position
        \emph{and} the three states $\ket{\alpha_1^{(m)}}, \ket{\alpha_2^{(m)}},
        \ket{\alpha_3^{(m)}}$ are linearly dependent but in general position.
        In this case, if we apply Lemma~\ref{lemma:gp}
        we get that the three states $\ket{\alpha_1^{([k] \setminus \{\ell\})}},
        \ket{\alpha_2^{([k] \setminus \{\ell\})}},
        \ket{\alpha_3^{([k] \setminus \{\ell\})}}$ are linearly independent,
        so $\dim A^{([k] \setminus \{\ell\})}_\tsep = 3$.

        We can therefore analyse the entanglement with respect to
        the bipartite partition $(\{\ell\} ~ , ~ [k] \setminus \{\ell\})$:
        the dimensions of the relevant local subspaces are
        $\dim A^{(\ell)}_\tsep = 2$
        (where the spanning states are in general position)
        and $\dim A^{([k] \setminus \{\ell\})}_\tsep = 3$,
        so the relevant case is (2,3)-i.
        Therefore, applying Theorem~\ref{thm:main} (case (2,3)-i)
        to $S$ under this partition
        implies that we are in case~\ref{multi_triangle}
        (a triangle), identically to three of the cases above.
        \item There is exactly \emph{one} subsystem $\ell \in [k]$ such that
        the three states $\ket{\alpha_1^{(\ell)}}, \ket{\alpha_2^{(\ell)}},
        \ket{\alpha_3^{(\ell)}}$ are linearly dependent but in general position.
        Thus, for all subsystems $m \in [k] \setminus \{\ell\}$,
        two of the three states $\ket{\alpha_1^{(m)}}, \ket{\alpha_2^{(m)}},
        \ket{\alpha_3^{(m)}}$ are equal to one another; formally, for each
        $m \in [k] \setminus \{\ell\}$ there exists a two-element subset
        $I_m \triangleq \{i_m, j_m\} \subset \{1,2,3\}$ such that
        $\ket{\alpha_{i_m}^{(m)}} = \ket{\alpha_{j_m}^{(m)}}$.
        We divide into two subcases:
        \begin{enumerate}
        \item If there exist two subsets $I_m$ and $I_n$
        which are different from one another (formally, if there exist
        $m, n \in [k] \setminus \{\ell\}$ such that $I_m \ne I_n$),
        then the three states $\ket{\alpha_1^{(m)}} \ket{\alpha_1^{(n)}} ~ , ~
        \ket{\alpha_2^{(m)}} \ket{\alpha_2^{(n)}} ~ , ~
        \ket{\alpha_3^{(m)}} \ket{\alpha_3^{(n)}}$ must be linearly independent.
        Therefore, according to Lemma~\ref{lemma:gp},
        the three states $\ket{\alpha_1^{([k] \setminus \{\ell\})}},
        \ket{\alpha_2^{([k] \setminus \{\ell\})}},
        \ket{\alpha_3^{([k] \setminus \{\ell\})}}$
        are, too, linearly independent,
        so $\dim A^{([k] \setminus \{\ell\})}_\tsep = 3$.

        We can therefore analyse the entanglement with respect to
        the bipartite partition $(\{\ell\} ~ , ~ [k] \setminus \{\ell\})$:
        the dimensions of the relevant local subspaces are
        $\dim A^{(\ell)}_\tsep = 2$
        (where the spanning states are in general position)
        and $\dim A^{([k] \setminus \{\ell\})}_\tsep = 3$,
        so the relevant case is (2,3)-i.
        Therefore, applying Theorem~\ref{thm:main} (case (2,3)-i)
        to $S$ under this partition
        implies that we are in case~\ref{multi_triangle}
        (a triangle), identically to four of the cases above.
        \item If all subsets $I_n$ are identical to each other
        for all $n \in [k] \setminus \{\ell\}$,
        then we can assume, without loss of generality, $I_n = \{2,3\}$.
        Let us choose an arbitrary $m \in [k] \setminus \{\ell\}$
        (so $\ket{\alpha_2^{(m)}} = \ket{\alpha_3^{(m)}}$),
        then, because $\ket{\alpha_1^{(\ell)}}, \ket{\alpha_2^{(\ell)}},
        \ket{\alpha_3^{(\ell)}}$ are in general position,
        the three states $\ket{\alpha_1^{(\ell)}}\ket{\alpha_1^{(m)}} ~ , ~
        \ket{\alpha_2^{(\ell)}}\ket{\alpha_2^{(m)}} ~ , ~
        \ket{\alpha_3^{(\ell)}}\ket{\alpha_3^{(m)}}$ are linearly independent
        according to Lemma~\ref{lemma:gp}.
        This implies that $\dim A^{(\{\ell,m\})}_\tsep = 3$.
        On the other hand, for all $n \in [k] \setminus \{\ell,m\}$ it holds that
        $\ket{\alpha_2^{(n)}} = \ket{\alpha_3^{(n)}}$, which in particular
        implies $\ket{\alpha_2^{([k] \setminus \{\ell,m\})}} =
        \ket{\alpha_3^{([k] \setminus \{\ell,m\})}}$.
        Therefore, $\dim A^{([k] \setminus \{\ell,m\})}_\tsep = 2$.

        We can therefore analyse the entanglement with respect to
        the bipartite partition $(\{\ell,m\} ~ , ~ [k] \setminus \{\ell,m\})$:
        the dimensions of the relevant local subspaces are
        $\dim A^{(\{\ell,m\})}_\tsep = 3$
        and $\dim A^{([k] \setminus \{\ell,m\})}_\tsep = 2$
        (where $\ket{\alpha_2^{([k] \setminus \{\ell,m\})}} =
        \ket{\alpha_3^{([k] \setminus \{\ell,m\})}}$
        are \emph{not} in general position),
        so the relevant case is (3,2)-ii.

        Therefore, applying Theorem~\ref{thm:main} (case (3,2)-ii)
        to $S$ under this partition
        implies that we are in case~\ref{multi_cone} (a spherical cone):
        the fully separable pure and mixed states over $S$
        are exactly all mixtures (convex combinations) of
        $\proj{\alpha_1^{(1)}} \ox \cdots \ox \proj{\alpha_1^{(k)}}$
        with any pure or mixed state over the space
        $\Span\left\{\ket{\alpha_2^{(\{\ell,m\})}},
        \ket{\alpha_3^{(\{\ell,m\})}}\right\}
        \ox \left\{\ket{\alpha_2^{([k] \setminus \{\ell,m\})}}\right\}
        = \Span\left\{\ket{\alpha_2^{(\ell)}}, \ket{\alpha_3^{(\ell)}}\right\}
        \ox \left\{\ket{\alpha_2^{([k] \setminus \{\ell\})}}\right\}$
        (this equality holds because $\ket{\alpha_2^{(m)}} =
        \ket{\alpha_3^{(m)}}$, and it proves that all those states
        are indeed fully separable),
        and all the other pure and mixed states over $S$ are entangled.
        \end{enumerate}
        \item There are \emph{no} subsystems $\ell \in [k]$ such that
        the three states $\ket{\alpha_1^{(\ell)}}, \ket{\alpha_2^{(\ell)}},
        \ket{\alpha_3^{(\ell)}}$ are in general position.
        Thus, for all subsystems $\ell \in [k]$,
        two of the three states $\ket{\alpha_1^{(\ell)}},
        \ket{\alpha_2^{(\ell)}}, \ket{\alpha_3^{(\ell)}}$
        are equal to one another; formally, for each
        $\ell \in [k]$ there exists a two-element subset
        $I_\ell \triangleq \{i_\ell, j_\ell\} \subset \{1,2,3\}$ such that
        $\ket{\alpha_{i_\ell}^{(\ell)}} = \ket{\alpha_{j_\ell}^{(\ell)}}$.

        The sets $I_\ell$ cannot be all identical to one another
        (for all $\ell \in [k]$), because then we would get $\dim S_\tsep = 2$.
        Therefore, there must exist $\ell, m \in [k]$ such that $I_\ell \ne I_m$.
        This means that the three states
        $\ket{\alpha_1^{(\ell)}} \ket{\alpha_1^{(m)}} ~ , ~
        \ket{\alpha_2^{(\ell)}} \ket{\alpha_2^{(m)}} ~ , ~
        \ket{\alpha_3^{(\ell)}} \ket{\alpha_3^{(m)}}$ are linearly independent,
        which means that $\dim A^{(\{\ell,m\})}_\tsep = 3$.
        We now divide into two subcases:
        \begin{enumerate}
        \item If there exist two subsystems $n,o \in [k] \setminus \{\ell, m\}$
        such that $I_n \ne I_o$, then the three states
        $\ket{\alpha_1^{(n)}} \ket{\alpha_1^{(o)}} ~ , ~
        \ket{\alpha_2^{(n)}} \ket{\alpha_2^{(o)}} ~ , ~
        \ket{\alpha_3^{(n)}} \ket{\alpha_3^{(o)}}$ are linearly independent,
        which means that $\dim A^{([k] \setminus \{\ell,m\})}_\tsep = 3$.

        We can therefore analyse the entanglement with respect to
        the bipartite partition $(\{\ell,m\} ~ , ~ [k] \setminus \{\ell,m\})$:
        the dimensions of the relevant local subspaces are
        $\dim A^{(\{\ell,m\})}_\tsep =
        \dim A^{([k] \setminus \{\ell,m\})}_\tsep = 3$,
        so the relevant case is (3,3).
        Therefore, applying Theorem~\ref{thm:main} (case (3,3))
        to $S$ under this partition
        implies that we are in case~\ref{multi_triangle}
        (a triangle), identically to five of the cases above.
        \item If all subsets $I_n$ are identical to each other
        for all $n \in [k] \setminus \{\ell, m\}$,
        then we can assume, without loss of generality, $I_n = \{2,3\}$.
        Then, for all $n \in [k] \setminus \{\ell, m\}$ it holds that
        $\ket{\alpha_2^{(n)}} = \ket{\alpha_3^{(n)}}$, which in particular
        implies $\ket{\alpha_2^{([k] \setminus \{\ell, m\})}} =
        \ket{\alpha_3^{([k] \setminus \{\ell, m\})}}$.
        Therefore, $\dim A^{([k] \setminus \{\ell, m\})}_\tsep = 2$.

        We can therefore analyse the entanglement with respect to
        the bipartite partition $(\{\ell, m\} ~ , ~ [k] \setminus \{\ell, m\})$:
        the dimensions of the relevant local subspaces are
        $\dim A^{(\{\ell,m\})}_\tsep = 3$
        and $\dim A^{([k] \setminus \{\ell,m\})}_\tsep = 2$
        (where $\ket{\alpha_2^{([k] \setminus \{\ell,m\})}} =
        \ket{\alpha_3^{([k] \setminus \{\ell,m\})}}$
        are \emph{not} in general position),
        so the relevant case is (3,2)-ii.

        Therefore, applying Theorem~\ref{thm:main} (case (3,2)-ii)
        to $S$ under this partition,
        we find that the separable pure and mixed states over $S$
        \emph{with respect to the bipartite partition
        $(\{\ell, m\} ~ , ~ [k] \setminus \{\ell, m\})$}
        (notice that these states are not necessarily fully separable!)
        are exactly all mixtures (convex combinations) of
        $\proj{\alpha_1^{(1)}} \ox \cdots \ox \proj{\alpha_1^{(k)}}$
        with any pure or mixed state over the space
        $\Span\left\{\ket{\alpha_2^{(\{\ell,m\})}},
        \ket{\alpha_3^{(\{\ell,m\})}}\right\}
        \ox \left\{\ket{\alpha_2^{([k] \setminus \{\ell,m\})}}\right\}
        = \Span\left\{\ket{\alpha_2^{(\ell)}}\ket{\alpha_2^{(m)}} ~ , ~
        \ket{\alpha_3^{(\ell)}}\ket{\alpha_3^{(m)}}\right\}
        \ox \left\{\ket{\alpha_2^{([k] \setminus \{\ell,m\})}}\right\}$,
        and all the other pure and mixed states over $S$
        are entangled with respect to this partition
        (and, thus, certainly not fully separable).

        Now we divide into two subsubcases:
        \begin{enumerate}
        \item If $I_n$ is equal to either $I_\ell$ or $I_m$
        (for all $n \in [k] \setminus \{\ell, m\}$),
        then without loss of generality, we can assume $I_n = I_m$.
        Therefore, $I_m = I_n = \{2, 3\}$,
        which means that $\ket{\alpha_2^{(m)}} = \ket{\alpha_3^{(m)}}$.

        The partition-specific result we reached above
        thus means that we are in case~\ref{multi_cone} (a spherical cone):
        the fully separable pure and mixed states over $S$
        are exactly all mixtures (convex combinations) of
        $\proj{\alpha_1^{(1)}} \ox \cdots \ox \proj{\alpha_1^{(k)}}$
        with any pure or mixed state over the space
        $\Span\left\{\ket{\alpha_2^{(\{\ell,m\})}},
        \ket{\alpha_3^{(\{\ell,m\})}}\right\}
        \ox \left\{\ket{\alpha_2^{([k] \setminus \{\ell,m\})}}\right\}
        = \Span\left\{\ket{\alpha_2^{(\ell)}}, \ket{\alpha_3^{(\ell)}}\right\}
        \ox \left\{\ket{\alpha_2^{([k] \setminus \{\ell\})}}\right\}$
        (this equality holds because $\ket{\alpha_2^{(m)}} =
        \ket{\alpha_3^{(m)}}$, and it proves that all those states
        are indeed fully separable),
        and all the other pure and mixed states over $S$ are entangled.
        \item If $I_\ell \ne I_n \ne I_m$
        (for all $n \in [k] \setminus \{\ell, m\}$),
        then in particular $I_\ell \ne \{2, 3\}$ and $I_m \ne \{2, 3\}$,
        which means that both $\ket{\alpha_2^{(\ell)}}, \ket{\alpha_3^{(\ell)}}$
        and $\ket{\alpha_2^{(m)}}, \ket{\alpha_3^{(m)}}$
        are linearly independent.

        Therefore, we can reconsider the meaning of
        the partition-specific result we reached above,
        which says that the only states over $S$
        that \emph{may} be fully separable are the mixtures of
        $\proj{\alpha_1^{(1)}} \ox \cdots \ox \proj{\alpha_1^{(k)}}$
        with any pure or mixed state over the local qubit space
        $\Span\left\{\ket{\alpha_2^{(\ell)}}\ket{\alpha_2^{(m)}} ~ , ~
        \ket{\alpha_3^{(\ell)}}\ket{\alpha_3^{(m)}}\right\}
        \ox \left\{\ket{\alpha_2^{([k] \setminus \{\ell,m\})}}\right\}$,
        and all the other pure and mixed states over $S$ are entangled.
        The only states in that local qubit space that \emph{could} be entangled
        (that is, not fully separable) are the non-trivial superpositions
        $\left(a\ket{\alpha_2^{(\ell)}}\ket{\alpha_2^{(m)}}
        + b\ket{\alpha_3^{(\ell)}}\ket{\alpha_3^{(m)}}\right)
        \ox \ket{\alpha_2^{([k] \setminus \{\ell,m\})}}$
        with $a, b \ne 0$; and we can indeed see that these states \emph{are}
        entangled, since both $\ket{\alpha_2^{(\ell)}}, \ket{\alpha_3^{(\ell)}}$
        and $\ket{\alpha_2^{(m)}}, \ket{\alpha_3^{(m)}}$
        are linearly independent, so a local invertible map
        (which does not affect entanglement and separability) can easily map
        them to orthonormal states:
        $\ket{\alpha_2^{(\ell)}} \mapsto \ket{1^{(\ell)}}$ ~ , ~
        $\ket{\alpha_3^{(\ell)}} \mapsto \ket{2^{(\ell)}}$ ~ , ~
        $\ket{\alpha_2^{(m)}} \mapsto \ket{1^{(m)}}$, and
        $\ket{\alpha_3^{(m)}} \mapsto \ket{2^{(m)}}$.
        The resulting state is thus
        $\left(a\ket{1^{(\ell)}}\ket{1^{(m)}}
        + b\ket{2^{(\ell)}}\ket{2^{(m)}}\right)
        \ox \ket{\alpha_2^{([k] \setminus \{\ell,m\})}}$,
        which is separable if and only if $a = 0$ or $b = 0$.
        Hence, the non-trivial superpositions in the local qubit space
        are indeed entangled.
        The remaining states are only the three states
        $\ket{\alpha_1^{(1)}} \ox \cdots \ox \ket{\alpha_1^{(k)}}$ ~ , ~
        $\ket{\alpha_2^{(1)}} \ox \cdots \ox \ket{\alpha_2^{(k)}}$, and
        $\ket{\alpha_3^{(1)}} \ox \cdots \ox \ket{\alpha_3^{(k)}}$,
        which are trivially product states.

        Thus, our analysis implies that we are in case~\ref{multi_triangle}
        (a triangle): the fully separable pure and mixed states over $S$
        are exactly all mixtures (convex combinations)
        of $\proj{\alpha_1^{(1)}} \ox \cdots \ox \proj{\alpha_1^{(k)}}$ ~ , ~
        $\proj{\alpha_2^{(1)}} \ox \cdots \ox \proj{\alpha_2^{(k)}}$,
        and $\proj{\alpha_3^{(1)}} \ox \cdots \ox \proj{\alpha_3^{(k)}}$,
        and all the other pure and mixed states over $S$ are entangled.
        \end{enumerate}
        \end{enumerate}
        \end{enumerate}
  \end{description}
  \vspace{-6.5mm}
\end{proof}

\section{Discussion}
\label{sec:discussion}
We have presented a full classification of entanglement and separability
in three-dimensional Hilbert subspaces of bipartite systems,
as well as a generalisation to multipartite systems.
This is a general classification that is independent of entanglement measures
and applies to all three-dimensional subspaces of any possible Hilbert space.
While our results are not aimed at deciding whether a single state is entangled
or separable (because that problem is easily solvable for rank-$3$ states
using the partial transpose criterion (PPT)~\cite{HLVC00},
and in some cases also using alternative criteria~\cite{Wootters}),
they are aimed at understanding the possible structures of entanglement and
separability in Hilbert spaces and their internal relations.

Our results generalise the findings of Boyer and two of the present authors
in~\cite{BLM}, which applied to two-dimensional subspaces.
In addition to the expected classes in three dimensions (mostly generalisations
or combinations of the classes from~\cite{BLM}),
we suggested an easy classification for each class
using the dimensions of local subspaces (along with information on whether
the spanning states of the subspaces are in general position, if applicable),
and we also found a few interesting classes that do not exist in two dimensions.

Our most interesting novel class, which is not similar to anything found in
two dimensions~\cite{BLM}, is named (2,2)-ii in Theorem~\ref{thm:main}
and is described in Fig.~\ref{fig:phi-x-phi}: it does not simply consist of
a finite number of product states or a complete Bloch sphere created by
degeneration of local eigenstates, like other classes,
but it includes all states of the form
$\ket{\psi}^\sA \ket{\psi}^\sB$ (up to local invertible operations)---that is,
all product states in the symmetric subspace (the triplet subspace),
as described in Section~\ref{sec:illustrations}.
The relation between the foundational phenomenon of the symmetric
subspace, appearing naturally here, and the possible classes of entanglement
in three dimensions and higher dimensions, is an intriguing topic
for future research.

Other possible directions for future research include extending our results
to four-dimensional subspaces
and higher-dimensional subspaces (where we may encounter more physical phenomena
not appearing here, including bound entanglement)
and finding more ways to utilise the geometric figures we suggested
in Section~\ref{sec:illustrations} for representing the sets of separable states
and their geometric features. 
It would also be interesting to know how the number of possible
classes of entanglement and separability (5 classes for two-dimensional
subspaces in~\cite{BLM}
and 14 classes for three-dimensional subspaces here) increases with the
dimension of the subspace, and in particular, whether the number of classes
is always finite.

One observation we can make based on Theorems~\ref{thm:BLM}
and~\ref{thm:main} is that for $2$- and $3$-dimensional subspaces, the number of
(pure) product states is either at most the dimension of the subspace
($2$ or $3$, respectively) or infinite
(cardinality of the continuum). It would be particularly interesting
to investigate the maximum finite number of pure product states in general
subspaces of arbitrary given dimension.

We have also generalised the results to multipartite systems. We point out that,
unlike the two-dimensional subspaces discussed in~\cite{BLM} where the bipartite
and multipartite cases are essentially equivalent, in three-dimensional
subspaces we found substantive differences between bipartite and multipartite
systems: most importantly, some of the classes for set of separable states
existing in bipartite systems completely disappear for (genuinely) multipartite
systems---including the interesting symmetric class (2,2)-ii mentioned above.
It could be interesting to find out why these classes disappear
in the multipartite case,
whether their disappearance hints at a foundational phenomenon
of quantum entanglement,
and whether similar differences between bipartite and multipartite
appear in higher dimensions, too.

\backmatter

\bmhead{Acknowledgements}

The authors thank Itai Arad for useful initial discussions on generalising
the results of \cite{BLM},
Ajit Iqbal Singh for discussions
on possible generalisations of the present work,
Kentaro Moto for insightful suggestions regarding identity and representations,
and the anonymous referees for significant additions to the text and
the references.
The work of RL and TM was partly supported
by the Israeli MOD Research and Technology Unit
and the Technion's Helen Diller Quantum Center (Haifa, Israel).
The work of RL was also partly supported by the Government of Spain 
(FIS2020-TRANQI and Severo Ochoa CEX2019-000910-S), Fundaci\'o Cellex, 
Fundaci\'o Mir-Puig, Generalitat de Catalunya (CERCA program), and the 
European Union NextGenerationEU.
AW is supported by the European Commission QuantERA grant ExTRaQT 
(Spanish MCIN project PCI2022-132965), by the Spanish MCIN
(project PID2019-107609GB-I00) with the support of FEDER funds, 
the Generalitat de Catalunya (project 2017-SGR-1127), by the Spanish 
MCIN with funding from European Union NextGenerationEU (PRTR-C17.I1) 
and the Generalitat de Catalunya, by the Alexander von Humboldt Foundation, 
and by the Institute for Advanced Study of the Technical University Munich. 

\bmhead{Data Availability Statement}
This manuscript has no associated data.

\section*{Declarations}

\bmhead{Conflict of interest}
The authors confirm they have no conflict of interest.

\bibliography{hgeometry}

\end{document}